\documentclass[pra,twocolumn,a4paper,showpacs,aps,10pt,nofootinbib,allowtoday,accepted=2019-12-31]{quantumarticle}
\pdfoutput=1
\usepackage{graphicx}
\usepackage{amsmath}
\usepackage{amsthm}
\usepackage{amssymb}
\usepackage{color}
\usepackage{comment}
\usepackage{framed}
\usepackage{appendix}

\newtheorem*{theorem*}{Theorem}
\newtheorem*{lemma*}{Lemma}
\newtheorem*{corollary*}{Corollary}

\newcommand{\Tr}{\text{Tr}}

\newtheorem{theorem}{Theorem}
\newtheorem{lemma}{Lemma}
\newtheorem{corollary}{Corollary}

\begin{document}

\title{Hypergraph framework for irreducible noncontextuality inequalities from logical proofs of the Kochen-Specker theorem}
\author{Ravi Kunjwal}

\email{quaintum.research@gmail.com}
\affiliation{Perimeter Institute for Theoretical Physics,\\ 31 Caroline Street North, Waterloo, Ontario, Canada, N2L 2Y5,}
\affiliation{Centre for Quantum Information and Communication, \'Ecole polytechnique de Bruxelles, 
	CP 165, Universit\'e libre de Bruxelles, 1050 Brussels, Belgium.}
\orcid{0000-0002-3978-5971}

\date{\today}                                           

\begin{abstract}
Kochen-Specker (KS) theorem reveals the inconsistency between quantum theory and any putative underlying model of it satisfying the constraint of {\em KS-noncontextuality}.  A logical proof of the KS theorem is one that relies {\em only} on the compatibility relations amongst a set of projectors (a {\em KS set}) to witness this inconsistency. These compatibility relations can be represented by a hypergraph, referred to as a {\em contextuality scenario}. Here we consider contextuality scenarios that we term {\em KS-uncolourable}, e.g., those which appear in logical proofs of the KS theorem. We introduce a hypergraph framework to obtain noise-robust witnesses of contextuality from such scenarios. 

Our approach builds on the results of R.~Kunjwal and R.~W.~Spekkens, \href{https://doi.org/10.1103/PhysRevLett.115.110403}{Phys.~Rev.~Lett.~115, 110403 (2015)}, by providing new insights into the relationship between the structure of a contextuality scenario and the associated noise-robust noncontextuality inequalities that witness contextuality.  
The present work also forms a necessary counterpart to the framework presented in R.~Kunjwal, \href{https://doi.org/10.22331/q-2019-09-09-184}{Quantum 3, 184 (2019)}, which only applies to {\em KS-colourable} contextuality scenarios, i.e., those which do not admit logical proofs of the KS theorem but do admit statistical proofs. 

We rely on a single hypergraph invariant, defined in R.~Kunjwal, \href{https://doi.org/10.22331/q-2019-09-09-184}{Quantum 3, 184 (2019)}, that appears in our contextuality witnesses, namely, the {\em weighted max-predictability}. The present work can also be viewed as a study of this invariant. Significantly, unlike the case of R.~Kunjwal, \href{https://doi.org/10.22331/q-2019-09-09-184}{Quantum 3, 184 (2019)}, {\em none} of the graph invariants from the graph-theoretic framework for KS-contextuality due to Cabello, Severini, and Winter (the ``CSW framework",\href{https://doi.org/10.1103/PhysRevLett.112.040401}{~Phys.~Rev.~Lett.~112, 040401 (2014)}) are relevant for our noise-robust noncontextuality inequalities.
\end{abstract}

\maketitle
\tableofcontents{}

\section{Introduction}
The Kochen-Specker (KS) theorem \cite{KS67} stands out as a fundamental insight into the nature of quantum measurements, formalizing the fact that these measurements cannot always be understood as merely revealing pre-existing values of physical quantities. However, the relevance of the KS theorem for real-world physics with finite-precision measurements has been a subject of intense controversy in the past \cite{meyer99,kent99,cliftonkent00,barrettkent04}. Recent work \cite{KunjSpek, exptlpaper, peresmerminksw, pusey, KunjSpek17, schmidWS} has taken the first steps towards turning the insight of the Kochen-Specker theorem into operational constraints --- or noncontextuality inequalities --- that are robust to noise and therefore experimentally testable. These inequalities do not presume the structure of quantum measurements --- in particular, that they are projective --- in their derivation, relying only on operational constraints that can be verified in an experiment and make sense, in particular, for nonprojective measurements in quantum theory.\footnote{Note that, strictly speaking, the traditional assumption of KS-noncontextuality \cite{KS67, kunjwalthesis} {\em can} be applied to arbitrary measurements if one doesn't care to justify outcome determinism from noncontextuality \`a la Ref.~\cite{Spe05}, but even so, it fails to make sense as a notion of classicality for nonprojective measurements in quantum theory, e.g., for trivial POVMs. See Section I and Appendix A of Ref.~\cite{robustcsw} for a discussion of this pathology of KS-noncontextuality as a notion of classicality. Also see Refs.~\cite{odum, finegen} for a more in-depth discussion of these issues, particularly the status of Fine's theorem \cite{fine} in the case of noncontextuality.} They are grounded in the generalized framework for contextuality proposed by Spekkens \cite{Spe05}. This framework is motivated by a methodological principle underlying the Kochen-Specker theorem: namely, noncontextuality as an application of the principle of ontological identity of operational indiscernables, or as Spekkens sometimes calls it, {\em Leibnizianity} \cite{Spekkenstalk,leibnizpaper}.

Recent developments in this operational approach to noncontextuality \`a la Spekkens have led to a plethora of noise-robust noncontextuality inequalities. We know examples of these for geometric KS constructions such as the 18 ray construction due to Cabello, Estebaranz, and Garc\'ia-Alcaine (CEGA)~\cite{KunjSpek,Cabello18ray} as well as algebraic constructions due to Peres and Mermin \cite{peresmerminksw,peres,mermin}. We also have novel examples that have no analogue in the traditional Kochen-Specker paradigm, such as the fair coin flip inequality \cite{exptlpaper}, the robust noncontextuality inequalities due to Pusey \cite{pusey} for the simplest nontrivial scenario admitting contextuality that is possible in the Spekkens approach, and the algorithmic approach to handling prepare-and-measure experiments with specified operational equivalences \cite{schmidWS}.

In this paper, which builds upon the conceptual ideas outlined in Ref.~\cite{KunjSpek}, we provide a hypergraph framework for obtaining noise-robust noncontextuality inequalities for any Kochen-Specker construction that is KS-uncolourable, i.e., {\em logical} proofs \cite{KunjSpek17} of the Kochen-Specker theorem in the style of Ref.~\cite{KS67}. The case of statistical proofs of the Kochen-Specker theorem --- such as the argument due to Klyachko et al.~\cite{KCBS} and the general graph-theoretic framework of Ref.~\cite{CSW} due to Cabello, Severini, and Winter (CSW) --- has already been formalized within the Spekkens framework \cite{Spe05} in two previous contributions: Ref.~\cite{KunjSpek17} provides the conceptual argument leading to this formalization while Ref.~\cite{robustcsw} provides a hypergraph-theoretic framework that obtains noise-robust noncontextuality inequalities from arbitrary statistical proofs of the KS theorem by going beyond the CSW framework \cite{CSW}. The approach we develop in this paper allows one to identify the physical quantities that one can expect to be constrained by the assumption of noncontextuality (and why) in KS-uncolourable contextuality scenarios. This is in contrast to the approach adopted in Refs.~\cite{KunjSpek,peresmerminksw,schmidWS} which require explicit enumeration of all the vertices of the polytope of probability assignments possible on the scenario and/or performing quantifier elimination. To circumvent this requirement, we rely on a characterization of extremal probabilistic models on contextuality scenarios proved in Ref.~\cite{AFLS} and a mapping from graphs to hypergraphs (for a family of contextuality scenarios) that we will introduce in this paper. 

The structure of the paper is as follows: In Section 2, we recall some definitions that will be used throughout the paper. In Section 3, we define a mapping from graphs to hypergraphs that can be used to generate KS-uncolourable hypergraphs; we obtain many known KS-uncolourable hypergraphs as examples. Section 4 defines a parameterization of contextuality scenarios that is used to obtain a characterization of KS-uncolourability for them. Section 5 recalls a theorem on extremal probabilistic models on contextuality scenarios, proved in Ref.~\cite{AFLS}, that will be crucial to obtaining our noncontextuality inequalities. In Section 6, we outline our general framework which is applicable to any KS-uncolourable contextuality scneario. In particular, Eq.~\eqref{ncineqschematic} provides the general form of noise-robust noncontextuality inequalities we obtain. We then use the results developed so far to obtain irreducible noncontextuality inequalities for some families of KS-uncolourable contextuality scenarios: in contrast to the inequality of Ref.~\cite{KunjSpek}, these inequalities cannot be reduced to any simpler ones. We prove a general theorem characterizing these inequalities for a whole family of KS-uncolourable contextuality scenarios. Table \ref{tabirrMISC} summarizes features of this family of contextuality scenarios that determine the structure of irreducible noncontextuality inequalities. Section 7 concludes with a summary of the framework and open questions that merit further research.

\section{Definitions}

\subsection{Operational theories}
We will be concerned with prepare-and-measure experiments in our tests of noncontextuality: that is, we imagine a preparation device as a source of a system that is subjected (following its preparation) to a measurement procedure carried out by a measurement device (see Fig.~\ref{fig1}).
\begin{figure}[htb!]
\centering
\includegraphics[scale=0.4]{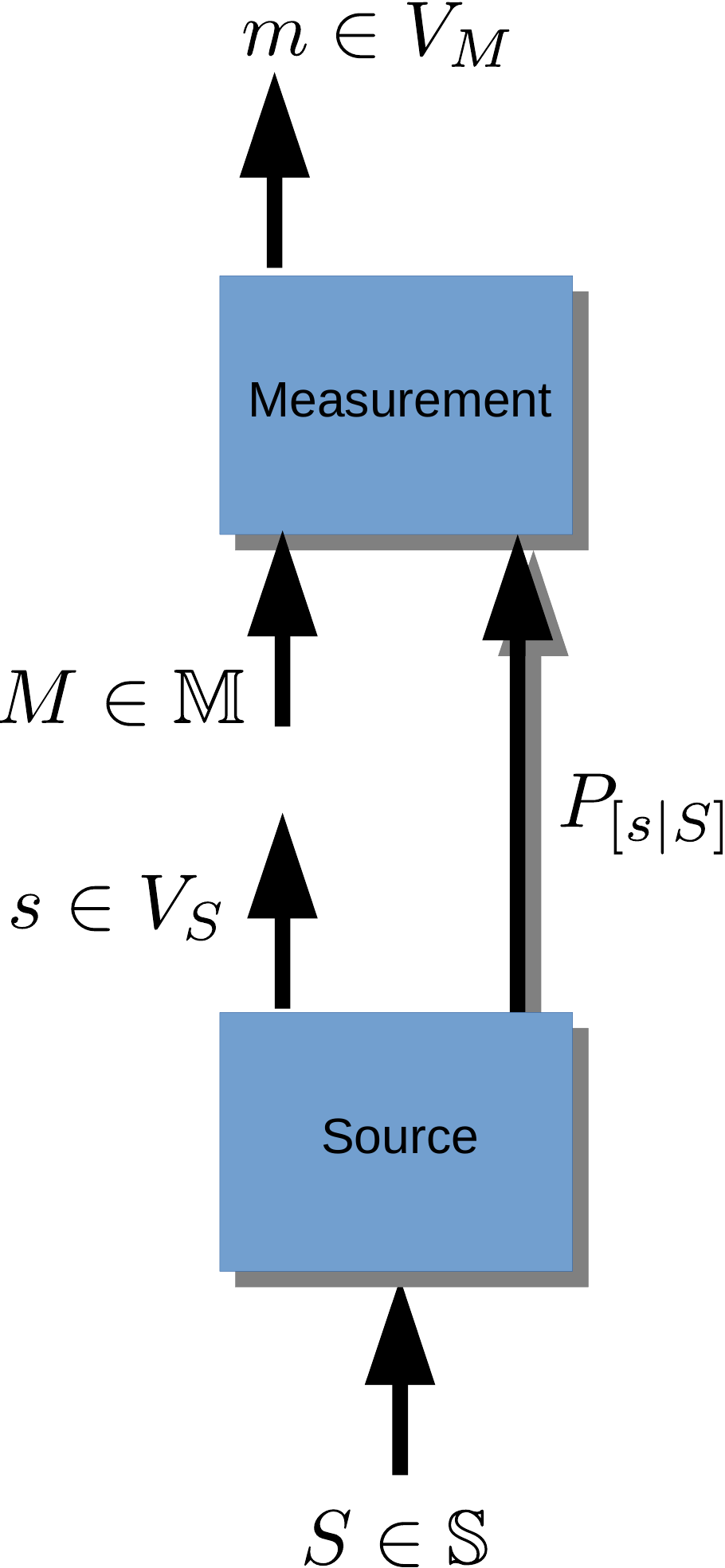}
\caption{A schematic of the prepare-and-measure experiment: the source setting $S\in\mathbb{S}$ produces a source outcome $s\in V_S$ and a system prepared according to preparation $P_{[s|S]}$ that is then subjected to the measurement device with measurement setting $M\in\mathbb{M}$ which then outputs a measurement outcome $m\in V_M$. The joint probability of source and measurement outcomes is given by $p(m,s|M,S)$. The source outcome $s$ occurs with probability $p(s|S)$ and setting $S$ prepares the ensemble $\{P_{[s|S]},p(s|S)\}_{s\in V_S}$. Time goes up: we assume that future settings/outcomes do not influence past settings/outcomes, so that $p(m,s|M,S)=p(m|M,S,s)p(s|M,S)=p(m|M,S,s)p(s|S)$.}
\label{fig1}
\end{figure}
The preparation device has many possible {\em source settings} $S\in\mathbb{S}$, each $S$ specifying a particular ensemble of preparation procedures labelled by (classical) {\em source outcomes} $s\in V_S$, where $V_S$ is the set of source outcomes for source setting $S$. We call $[s|S]$ a {\em source event}. Hence, each preparation procedure is denoted by $P_{[s|S]}$, corresponding to the source event $[s|S]$: that is, the system is prepared according to $P_{[s|S]}$ with probability $p(s|S)\in[0,1]$, where $\sum_{s\in V_S}p(s|S)=1$, for every choice of setting $S\in\mathbb{S}$.
The ensemble of preparation procedures associated with the source setting $S$ is then given by $\{P_{[s|S]},p(s|S)\}_{s\in V_S}$. 

Similarly, the measurement device has many possible {\em measurement settings} $M\in\mathbb{M}$, each $M$ specifying 
a particular measurement procedure with many possible {\em measurement outcomes} $m\in V_M$, where $V_M$ is the set of values the 
measurement device can output when the measurement setting is $M$ and a system prepared by some source is an input to the 
measurement device. 
We will use $[m|M]$ to denote the {\em measurement event} that the outcome $m$ was witnessed for measurement setting $M$.

The joint probability of a particular outcome $m$ for measurement setting $M$ and a particular outcome 
$s$ for source setting $S$ when the input to the measurement device is a system prepared according to procedure $P_{[s|S]}$ 
is given by $p(m,s|M,S)\in[0,1]$, where $\sum_{m,s}p(m,s|M,S)=1$. One can view this joint probability as composed of two pieces:
$$p(m,s|M,S)=p(m|M,S,s)p(s|S),$$
where $p(m|M,S,s)$ is the conditional probability of outcome $m$ for measurement $M$ when the system is prepared according to 
procedure $P_{[s|S]}$ and $p(s|S)$ is the conditional probability that the system is indeed prepared according to $P_{[s|S]}$ for the source (setting) $S$.

An operational theory is therefore a specification of the triple $(\mathbb{S},\mathbb{M},p)$, where $\mathbb{S}$ is the set of source settings
in the operational theory,
$\mathbb{M}$ is the set of measurement settings, and $p:V_M\times V_S\times\mathbb{M}\times\mathbb{S}\rightarrow[0,1]$ is a function that specifies the joint
probability $p(m,s|M,S)$ that a given source setting $S\in\mathbb{S}$ and measurement setting $M\in\mathbb{M}$ produce respective outcomes $s\in V_{S}$ and $m\in V_{M}$ when the system prepared by the preparation device is fed to the measurement device, 
and where $\sum_{m,s}p(m,s|M,S)=1$ for all $M\in\mathbb{M}$, $S\in\mathbb{S}$.

\subsection{Ontological models}
An ontological model of an operational theory seeks to provide an explanatory framework for its predictions, grounding them in 
intrinsic properties of physical systems. 
All such properties
of a system are presumed to be encoded in its {\em ontic state} $\lambda\in\Lambda$, where $\Lambda$ is the set of all possible ontic states 
of the system. A source event $[s|S]$ prepares the system in ontic state $\lambda$ with probability $\mu(\lambda|S,s)\in[0,1]$,
where $\sum_{\lambda}\mu(\lambda|S,s)=1$. 
On measuring the system in ontic state $\lambda$, the measurement $M$ produces outcome $m$ with
probability $\xi(m|M,\lambda)\in[0,1]$, where $\sum_m\xi(m|M,\lambda)=1$ for all $\lambda\in\Lambda$. We then have:
\begin{equation}\label{empadeq1}
 p(m|M,S,s)=\sum_{\lambda\in\Lambda}\xi(m|M,\lambda)\mu(\lambda|S,s).
\end{equation}
We can use Bayes' theorem to write $\mu(s|S,\lambda)=\frac{\mu(s,\lambda|S)}{\mu(\lambda|S)}$.
Noting that $\mu(s,\lambda|S)=\mu(\lambda|S,s)p(s|S)$, we have $$\mu(s|S,\lambda)=\frac{\mu(\lambda|S,s)p(s|S)}{\mu(\lambda|S)},\textrm{ or}$$
\begin{equation}
 \mu(\lambda|S,s)=\frac{\mu(s|S,\lambda)\mu(\lambda|S)}{p(s|S)}.
\end{equation}
Substituting this in Eq.~(\ref{empadeq1}), we have
\begin{equation}\label{empadeq}
 p(m,s|M,S)=\sum_{\lambda\in\Lambda}\xi(m|M,\lambda)\mu(s|S,\lambda)\mu(\lambda|S),
\end{equation}
which describes how the operational joint probabilities of the prepare-and-measure experiment must be reproduced by the ontological 
model. Here $\xi(m|M,\lambda)$ is the {\em predictive} probability that a particular outcome $m$ {\em will occur} for a given measurement setting $M$
when the input ontic state is $\lambda$ while $\mu(s|S,\lambda)$ is the {\em retrodictive} probability that a particular outcome $s$ {\em occurred}
for a given source setting $S$ which produced the ontic state $\lambda$. $\mu(\lambda|S)$ is the probability that $\lambda$ was sampled by the source setting $S$ at all, ignoring its outcomes $s\in V_S$.

\subsection{Operational equivalences and Noncontextuality}
\subsubsection{Operational equivalences}Two source events $[s|S]$ and $[s'|S']$ are said to be operationally equivalent,
denoted $[s|S]\simeq [s'|S']$, if:
\begin{align}
\forall [m|M]&: p(m,s|M,S)=p(m,s'|M,S'),\nonumber\\
&\textrm{ where } M\in\mathbb{M}, m\in V_M.
\end{align}
Two source settings $S$ and $S'$ are said to be operationally equivalent, denoted $[\top|S]\simeq [\top|S']$, if:
\begin{align}
 \forall [m|M]&:\sum_{s\in V_S}p(m,s|M,S)=\sum_{s'\in V_{S'}}p(m,s'|M,S'),\nonumber\\
 &\textrm{ where }M\in\mathbb{M}, m\in V_M.
\end{align}
The symbol ``$\top$" denotes coarse-graining over all outcomes, i.e., the $[\top|S]$ is the source event that at least one outcome in $V_S$ occurred for source setting $S$. In this paper, we will only make use of such coarse-grained operational equivalences between source settings, that is, ones where we sum over the classical outcomes of the sources.

Two measurement events $[m|M]$ and $[m'|M']$ are said to be operationally equivalent, denoted $[m|M]\simeq[m'|M']$, if:
\begin{align}
 \forall [s|S]&:p(m,s|M,S)=p(m',s|M',S),\nonumber\\
 &\textrm{ where }S\in\mathbb{S}, s\in V_S.
\end{align}

Note that, because of normalization, any two coarse-grained measurement settings $M$ and $M'$ are always operationally equivalent, i.e.,
$\sum_{m\in V_M}p(m,s|M,S)=\sum_{m'\in V_{M'}}p(m',s|M',S)=p(s|S)$ for all $[s|S]$. It's only in the case of sources that the operational equivalence after coarse-graining is nontrivial, i.e., it needs to be verified experimentally.

\subsubsection{Context}
Any distinction between operationally equivalent experimental procedures --- preparations or measurements --- is called a {\em context}.

\subsubsection{Noncontextuality}
We define the notion of noncontextuality following the proposal by Spekkens \cite{Spe05}, wherein
the assumption of noncontextuality requires operationally equivalent experimental procedures to be represented identically in the ontological
model. That is, differences of context between operationally equivalent experimental procedures should be as irrelevant in the ontological model
as they are in the operational theory. Indeed, this indifference to variations in context -- that is, {\em noncontextuality} -- in the ontological model is meant to account for the indifference to variations in context -- that is, {\em operational equivalence} -- that holds in the operational theory. Our goal is to put this hypothesis of noncontextuality to experimental test by figuring out operational constraints -- noncontextuality inequalities -- that it imposes on the operational statistics.

Thus, the assumption of preparation noncontextuality applied to operationally equivalent source events, $[s|S]\simeq [s'|S']$, reads
\begin{equation}
 \mu(\lambda|S,s)=\mu(\lambda|S',s')\quad\forall\lambda\in\Lambda.
\end{equation}
Applied to the operational equivalence $[\top|S]\simeq[\top|S']$, preparation noncontextuality reads
\begin{align}
 \forall\lambda\in\Lambda&:\sum_{s\in V_S}p(s|S)\mu(\lambda|S,s)\nonumber\\
 &=\sum_{s'\in V_{S'}}p(s'|S')\mu(\lambda|S',s'),\textrm{ or}
\end{align}
\begin{equation}
 \mu(\lambda|S)=\mu(\lambda|S')\quad\forall\lambda\in\Lambda,
\end{equation}
where $$\mu(\lambda|S)\equiv\sum_{s\in V_S}\mu(s,\lambda|S)=\sum_{s\in V_S}p(s|S)\mu(\lambda|S,s),$$ and similarly 
$$\mu(\lambda|S')\equiv\sum_{s'\in V_{S'}}\mu(s',\lambda|S')=\sum_{s'\in V_{S'}}p(s'|S')\mu(\lambda|S',s').$$
We will only make use of 
this type of preparation noncontextuality in this paper.

The assumption of measurement noncontextuality applied to operationally equivalent measurement events $[m|M]\simeq[m'|M']$ reads
\begin{equation}
 \xi(m|M,\lambda)=\xi(m'|M',\lambda)\quad\forall\lambda\in\Lambda.
\end{equation}

\subsection{Contextuality scenarios and probabilistic models on them}
In keeping with the definitions of Ref.~\cite{AFLS}, we introduce the following notions:
\begin{itemize}
 \item {\em Contextuality scenario}: A contextuality scenario is a hypergraph $H$ where the nodes of the hypergraph $w\in W(H)$ denote measurement outcomes and 
hyperedges denote measurements $f\in F(H)\subseteq 2^{W(H)}$ such that $\bigcup_{f\in F(H)}=W(H)$. We will assume the set of nodes $W(H)$ is finite
and, therefore, so is the set of hyperedges $F(H)$.

A node shared between multiple hyperedges represents a measurement outcome with multiple possible measurement contexts in which it can occur. This is the notion of a (measurement) context that is used in logical proofs of the Kochen-Specker theorem relying on KS-uncolourability \cite{KS67,Cabello18ray}. We will be concerned with this notion of measurement context in this paper.\footnote{Spekkens contextuality \cite{Spe05} encompasses the measurement contexts relevant in the Kochen-Specker paradigm but does not restrict itself to them (see, e.g., \cite{exptlpaper}). In particular, it allows for a notion of preparation contexts which has been previously used in Ref.~\cite{KunjSpek} -- the conceptual precursor to the present work -- and which we will use in this paper.}

\item {\em $n$-hypercycle}: An $n$-hypercycle is a collection of $n$ nodes, $\{w_i\}_{i=1}^n$, in a hypergraph $H\equiv(W,F)$ such that 
for all $i\in\{1,2,\dots,n\}$, $\{w_i,w_{i+1}\}\subseteq f$ for some $f\in F$, but
no other subsets of $\{w_i\}_{i=1}^n$ appear in a hyperedge of $H$. Note that we have assumed addition modulo $n$, i.e., $n+1=1$, while labelling the nodes.

Note that every $n$-hypercycle in a hypergraph is also an $n$-cycle, where by ``$n$-cycle" we refer to a collection of $n$ nodes, $\{w_i\}_{i=1}^n$, such that 
for all $i\in\{1,2,\dots,n\}$, $\{w_i,w_{i+1}\}\subseteq f$ for some $f\in F$. On the other hand, not every $n$-cycle in a hypergraph is an $n$-hypercycle.

When the hypergraph $H$ is a graph, i.e., every hyperedge $f\in F$ contains exactly two nodes from $W$, then every $n$-cycle of $H$ is also an $n$-hypercycle.\footnote{
Note that we are not necessarily imagining that $H$ itself is isomorphic to an $n$-cycle graph. $H$ can be any arbitrary hypergraph and the question is if it admits subhypergraphs that are $n$-hypercycles. The $n$-hypercycles that are contained in $H$ are distinct from (and, in general, fewer in number than) the $n$-cycles contained in it. For example, consider the $6$-vertex hypergraph $H=\{\{1,2,3\},\{3,4,5\},\{5,6,1\}\}$. This hypergraph contains the following $n$-cycles: $\{\{1,2\},\{2,3\},\{1,3\}\}$, $\{\{3,4\},\{4,5\},\{3,5\}\}$, $\{\{5,6\},\{6,1\},\{1,5\}\}$, and $\{\{1,3\},\{3,5\},\{1,5\}\}$. However, only one of these $n$-cycles is an $n$-hypercycle, namely, $\{\{1,3\},\{3,5\},\{1,5\}\}$, since its vertices are not all contained in a single hyperedge. When the hypergraph is a graph, e.g., $H'=\{\{1,3\},\{3,5\},\{1,5\}\}$, then all its $n$-cycles are also $n$-hypercycles.
}

\item {\em Probabilistic model}: A probabilistic model on a contextuality scenario is an assignment of probabilities to the nodes of 
 a hypergraph, $p:W(H)\rightarrow[0,1]$, such that the hyperedges are normalized, i.e., $\sum_{w\in f}p(w)=1$ for all $f\in F(H)$. We denote the set of such general probabilistic models on $H$ by $\mathcal{G}(H)$.
 
Viewed operationally, any probabilistic model on the contextuality scenario specifies the probabilities of measurement outcomes when the measurements are implemented on some preparation in an operational theory. A given operational theory may only allow a certain subset of all possible probabilistic models on a contextuality scenario when the measurements in the scenario are implemented on a preparation possible in the operational theory. Indeed, that is the premise of Ref.~\cite{AFLS}, where possible probabilistic models on a given contextuality scenario are classified as {\em classical}, {\em quantum}, or {\em general probabilistic} models. The full set of possible probabilistic models on the contextuality scenario, corresponding to a polytope, constitutes the set of general probabilistic models. 

We will be interested in the polytope of general probabilistic models in this paper. In particular, we do {\em not} seek to classify probabilistic models on a contextuality scenario \`a la Acin, Fritz, Leverrier, and Sainz (AFLS) \cite{AFLS}. Instead, we will be interested in properties of these probabilistic models that become crucial {\em only} in the operational approach \`a la Spekkens \cite{Spe05}, having no analogue in the AFLS framework. This is to be expected since the AFLS framework is a formalization of the Kochen-Specker paradigm and we seek to ask questions that necessitate a reformulation and extension of this paradigm \`a la Spekkens. For the case of {\em statistical} proofs of the Kochen-Specker theorem, this has been achieved in Refs.~\cite{KunjSpek17, robustcsw}. This paper seeks to achieve this for {\em logical} proofs of the Kochen-Specker theorem, based on the ideas conceptualized in Ref.~\cite{KunjSpek}. Our goal here is to provide technical tools concerning the sorts of hypergraph properties -- distinct from the ones discussed in, for example, Refs.~\cite{CSW, AFLS} -- that are relevant for the noise-robust noncontextuality inequalities we derive in the operational approach \`a la Spekkens. These hypergraph properties are easily captured in a new hypergraph invariant that we defined in Ref.~\cite{robustcsw} --- the {\em weighted max-predictability} $\beta(\Gamma,q)$ for a contextuality scenario $\Gamma$ with hyperedges weighted by probabilities according to the distribution $q$ ---  and which we will define in due course in this paper. We will use $\beta(\Gamma,q)$ to obtain noise-robust noncontextuality inequalities arising from any contextuality scenario $\Gamma$ yielding a logical proof of the KS theorem.

Note that the probability assigned to a measurement outcome (or node) by {\em any} probabilistic model on the hypergraph does not vary with the measurement context (or hyperedge) that the measurement outcome may be considered a part of: operationally, this means that the operational theories that lead to various probabilistic models on contextuality scenarios exhibit nontrivial operational equivalences between measurement outcomes of different measurements. These operational equivalences take the form of the {\em same} node being shared between two (or more) hyperedges and are represented in their entirety by the structure of the hypergraph denoting the contextuality scenario. The assumption of measurement noncontextuality -- as we have defined it -- will be applied to these operational equivalences implicit in the contextuality scenario.

 \item {\em Kochen-Specker (KS) uncolourable scenario}: A KS-uncolourable scenario is a contextuality scenario which does not admit a probabilistic model that is deterministic, i.e., $p:W(H)\rightarrow\{0,1\}$, even though it may admit indeterministic probabilistic models.\footnote{For readers familiar with the notion of ``strong contextuality" in the sheaf-theoretic approach of Abramsky and Brandenburger \cite{AB}, we note that KS-uncolourability is a property of the contextuality scenario itself rather than (as in the case of ``strong contextuality") of a particular probabilistic model on it. In this sense, KS-uncolourability is distinct from {\em strong contextuality}. However, the two notions {\em are} related in the sense that KS-uncolourability of a contextuality scenario implies strong contextuality for {\em all} probabilistic models on it. Strong contextuality is the property of a probabilistic model: namely, that the probabilistic model does not admit a convex decomposition that has in its support a deterministic model(s). Extremal probabilistic models on a KS-colourable contextuality scenario that are indeterministic are, by definition, strongly contextual. On the other hand, {\em every} probabilistic model on a KS-uncolourable contextuality scenario (extremal or not) is strongly contextual.}
 
 A contextuality scenario which does admit deterministic probabilistic models is called {\em KS-colourable}.
 
 \item {\em Kochen-Specker (KS) set}: A KS set is a set of rank 1 projectors, $\{\Pi_w\}_{w\in W(H)}$, 
 on some Hilbert space $\mathcal{H}$ that can be associated with the nodes $W(H)$ of a KS-uncolourable scenario such that 
 $\sum_{w\in f}\Pi_w=\mathbb{I}$ for all $f\in F(H)$ and $\Pi_w\Pi_{w'}=\delta_{w,w'}\Pi_w$ for any $w,w'\in f$. 
 Here $\mathbb{I}$ is the identity operator on $\mathcal{H}$. 
 
 Each KS set corresponds to an infinity of possible probabilistic models on a KS-uncolourable contextuality scenario, each given by $p(w)=\Tr\rho\Pi_w$
 for all $w\in W(H)$, for some density operator $\rho$ on $\mathcal{H}$.

\item {\em (Induced) Subscenario:} Given a contextuality scenario $H$ with nodes $W(H)$ and contexts $F(H)$, the subscenario $H_S$ induced 
by a subset of nodes $S\subseteq W(H)$ is given by: $W(H_S)\equiv S$ and $F(H_S)\equiv \{f\cap S| f\in F(H)\}$.

That is, $H_S$ is constructed by dropping all the nodes in $W(H)\backslash S$ and restricting the hyperedges in $F(H)$ to their 
intersection with the nodes in $S$. 

{\em Remark on probabilistic models on $H_S$: }
Note that an induced subscenario $H_S$ admits a valid probabilistic model only if 
$f\cap S\neq \varnothing$ for all contexts $f\in F(H)$ because otherwise the set of hyperedges in $H_S$ will include 
empty sets which cannot be normalized, rendering a probabilistic model impossible on $H_S$ (that is, $\mathcal{G}(H_S)=\varnothing$).
In the language of hypergraph theory, $S$ must be a {\em transversal} (or ``hitting set'') of $H$.

\item {\em Extension of a probabilistic model on an induced subscenario ($H_S$) to the parent contextuality scenario ($H$)}: Every probabilistic model on $H_S$, say $p_S$, can be {\em extended} to a probabilistic model
$p$ on $H$ as follows: $p(w)=p_S(w)\quad\forall w\in S$ and $p(w)=0$ otherwise. $p$ is said to be an {\em extension} of $p_S$
from $H_S$ to $H$.\footnote{Note that the definition of an induced subscenario ($H_S$) and the extension of a probabilistic model from $H_S$ to $H$ is adopted from Ref.~\cite{AFLS}, specifically Definition 2.5.1 in Ref.~\cite{AFLS}.}
\end{itemize}

We now recall Theorem 2.5.3 of Ref.~\cite{AFLS}, a characterization of extremal probabilistic models on a contextuality scenario, that we will use:
\begin{theorem}\label{extremals}
(Theorem 2.5.3 in \cite{AFLS})\\
$p\in\mathcal{G}(H)$ is extremal if and only if it is the extension of a unique probabilistic model 
$p_S$ on an induced subscenario $H_S$ (that is, $\mathcal{G}(H_S)=\{p_S\}$) to the probabilistic model $p$ on $H$.
\end{theorem}
Hence, there is a one-to-one correspondence between extremal probabilistic models on $H$ and induced subscenarios of $H$ 
with unique probabilistic models. Indeed, as noted in \cite{AFLS}, this means that each extremal probabilistic model $p\in\mathcal{G}(H)$ is in one-to-one 
correspondence with the set of vertices assigned nonzero probability by the extremal probabilistic model, i.e. $S_p\equiv\{w\in W(H)|p(w)\neq 0\}$.\footnote{This fact seems to be closely related to the later work of Abramsky {\em et al.}\cite{nspolytopes}, where the combinatorial structure of no-signalling polytopes is characterized in entirely possibilistic terms. In particular, characterization of the face lattice of the polytope $\mathcal{G}(H)$ of probabilistic models on $H$ (as done in Ref.~\cite{nspolytopes}, albeit within a different formalism) implies a characterization of the vertices of this polytope (as done in Ref.~\cite{AFLS}). The proof of Theorem 2.5.3 in Ref.~\cite{AFLS} provides direct evidence of this connection.}

\section{A family of KS-uncolourable scenarios: the mapping 2Reg($\cdot$)}
We now consider a particular mapping, which we call 2Reg($\cdot$), that often converts a graph to a contextuality scenario that is KS-uncolourable.\footnote{``2Reg'' refers to the fact that contextuality scenarios obtained under this mapping are such that every node of the 
scenario has degree 2, i.e., it appears in two hyperedges or contexts, hence the scenario is 2-regular.}
We will see that many known examples of KS-uncolourable scenarios arise in this way. The mapping 2Reg($\cdot$) is defined in the following
manner:

\begin{itemize}
 \item Input Graph: $G=(V,E)$, where $V=\{v_1,v_2,\dots,v_{|V|}\}$ and $E=\{e_1,e_2,\dots,e_{|E|}\}\subseteq \{\{v_i,v_j\}|i,j\in\{1,\dots,n\}, i\neq j\}$.
 \item Output Hypergraph: 2Reg$(G)\equiv H=(W,F)$, where 
 \begin{itemize}
  \item  each node $w_k\in W$ ($k=1,\dots,|W|$) is defined by a pair of edges in $E$ that share a vertex, i.e.,
 for every pair $\{e_i,e_j\}\subseteq E$ (where $i\neq j$) such that $e_i\cap e_j\neq\varnothing$, there is a corresponding node
$w_k\in W$. The cardinality of $W$, $|W|$, is therefore equal to the number of distinct pairs of edges in $E$ such that 
each pair shares a vertex (in $V$).
\item For every edge $e_i\in E$, we define a corresponding hyperedge $f_i\in F$ such that the nodes in $f_i$ correspond 
precisely to the pairs of edges $\{\{e_i,e_j\}|e_i\cap e_j\neq\varnothing, j\neq i\text{ and } j\in\{1,\dots,n\}\}.$
We have $|F|=|E|$.
 \end{itemize}
\end{itemize}
A key property of such a hypergraph, $H$, generated via 2Reg$(G)$
is that each node $w_k$ (corresponding to a pair $\{e_i,e_j\}$, say) in $W$ appears in exactly two hyperedges ($f_i,f_j\in F$) 
of the hypergraph $H$. That is, 2Reg$(G)$ is a 2-regular hypergraph for any graph $G$. This is simply because every node in $H$ is 
essentially {\em defined} by the intersection of two hyperedges in $H$.\footnote{Note that if $G$ contains (at least) an edge disjoint from the rest of $G$, the mapping 2Reg($\cdot$) results in a set of hyperedges where (at least) one of them (corresponding to the disjoint edge) is empty. Strictly speaking, this is not a hypergraph since it has an empty hyperedge and thus 2Reg($\cdot$) will not result in a contextuality scenario. Therefore, such graphs $G$ will not be of interest to us. Henceforth, we will only consider $G$ for which 2Reg($G$) is a contextuality scenario.}

To see how this mapping works consider an example: the complete bipartite graph $K_{3,3}$ with vertices $V=\{1,2,3,\bar{1},\bar{2},\bar{3}\}$
and edges $E=\{(1\bar{1}),(1\bar{2}),(1\bar{3}),(2\bar{1}),(2\bar{2}),(2\bar{3}),(3\bar{1}),(3\bar{2}),(3\bar{3})\}$ transforms under 2Reg($\cdot$)
to the CEGA hypergraph \cite{Cabello18ray} which, given a realization with $18$ rays in $\mathbb{R}^4$, provides a proof of the KS theorem in $4$ dimensions. See Fig.~\ref{fig2}.
We will see how this mapping works in a forthcoming section on complete bipartite graphs under 2Reg($\cdot$).

\begin{figure}
\includegraphics[scale=0.38]{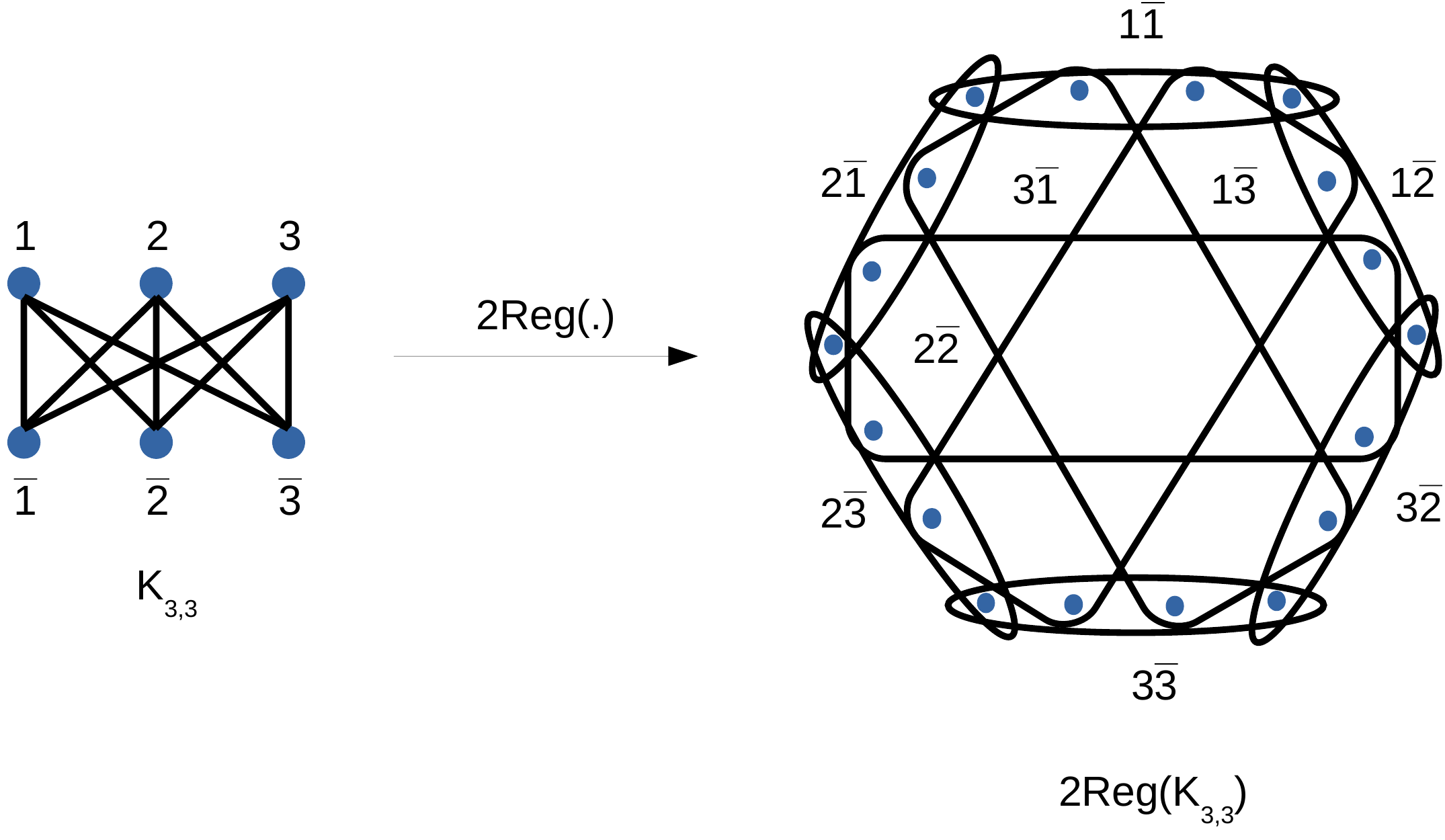}
\caption{The bipartite graph $K_{3,3}$ under the mapping 2Reg($\cdot$).}
\label{fig2}
\end{figure}

\subsection{KS-uncolourability under 2Reg($\cdot$)}
\begin{theorem}\label{thmunc}
Given a graph $G=(V,E)$, the contextuality scenario 2Reg$(G)$ is KS-uncolourable if and only if $|E|$ is odd.
\end{theorem}
\begin{proof}
$|E|$ is odd $\Rightarrow$ 2Reg$(G)$ is KS-uncolourable:
We have $|E|$ normalization equations for any KS-noncontextual assignment of $\{0,1\}$ probabilities to the nodes of the hypergraph 
2Reg$(G)$, since 2Reg$(G)$ has $|E|$ contexts and the $\{0,1\}$-valued assignments to the nodes in each context should
add up to 1; on adding all the normalization equations together,
 we note that since each $\{0,1\}$-valued assignment to a node appears in two different equations we have an even number on the 
 left-hand-side (LHS) of the resulting equation and an odd number ($=|E|$)
 on the right-hand-side (RHS), hence there is no $\{0,1\}$-valued solution to the normalization equations and 2Reg$(G)$
 is KS-uncolourable. 

To prove the converse, we show that
$|E|$ is even $\Rightarrow$ 2Reg$(G)$ is KS-colourable: Note that 2Reg$(G)$
 consists of a set of $|E|$ contexts such that every pair of them with a non-empty intersection shares exactly one node. Let's call these contexts 
 $C_1$, $C_2$, $C_3$,\dots,$C_{|E|}$, labelled such that $C_i$ and $C_{i+1}$ (addition modulo $|E|$, so $|E|+1=1$)
 share a node for all $i\in\{1,2,\dots,|E|\}$. Consider now the even hypercycle of size $|E|$ given by the contexts 
 $$C_1 - C_2 - C_3 - \dots - C_{|E|} - C_1$$ and assign the probability 1 to node defined by the intersection of $C_1$ and $C_2$ (denoted $C_1 - C_2$), probability 0 to 
 node defined by the intersection of $C_2$ and $C_3$ (denoted $C_2 - C_3$), 1 to node defined by intersection of $C_3$ and $C_4$ (denoted  $C_3 - C_4$),\dots, and so on, alternating assignments of 1 and 0,
 up to assigning probability 0 to the node denoted $C_{|E|} - C_1$.\footnote{Note that ``$C_i-C_{i+1}$'' denotes the node that 
 appears in both contexts, $C_i$ and $C_{i+1}$.} The induced subscenario consisting of singleton hyperedges (that is, hyperedges with a 
 single node each)
 $$\{\{C_1-C_2\},\{C_3-C_4\},\dots,\{C_{|E|-1}-C_{|E|}\}\}$$
 then admits a unique probabilistic model which extends to a deterministic extremal probabilistic model on 2Reg$(G)$.
 This is easy to see because the induced subscenario assigns probability 1 to $\frac{|E|}{2}$ nodes and each of those nodes 
 appears in two contexts, thus ensuring that all the contexts are properly normalized in the
 extension of the unique probabilistic model to 2Reg$(G)$. Hence, 2Reg$(G)$ is KS-colourable whenever $|E|$ is even.

\end{proof}

\subsection{From graphs to hypergraphs under 2Reg($\cdot$)}
We will now prove some facts about the behaviour of some special classes of graphs under 2Reg($\cdot$).
\begin{lemma}\label{ncycle}
All $n$-cycle ($n\geq3$) graphs are invariant under 2Reg($\cdot$).
\end{lemma}

\begin{proof}
Given the $n$-cycle graph 
 $$\{e_{12}\equiv\{v_1,v_2\},e_{23}\equiv\{v_2,v_3\},\dots,e_{n1}\equiv\{v_n,v_1\}\},$$ the hypergraph under 2Reg($\cdot$) is 
given by $$\{E_{12}\equiv\{w_1\equiv (e_{12},e_{n1}),w_2\equiv (e_{12},e_{23})\},$$$$E_{23}\equiv\{w_2\equiv (e_{12},e_{23}),w_3\equiv (e_{23},e_{34})\},\dots,
$$$$E_{n1}\equiv\{w_n\equiv (e_{n-1,n},e_{n1}),w_1\equiv (e_{12},e_{n1})\}\},$$
which is the $n$-hypercycle isomorphic to the starting $n$-cycle graph. From Theorem \ref{thmunc}, an $n$-hypercycle is 
KS-uncolourable if and only if $n$ is odd. Note, in particular, that a triangle (or 3-cycle) graph goes to a $3$-hypercycle. See Fig.~\ref{fig3}.
\end{proof}

\begin{figure}
\centering
\includegraphics[scale=0.5]{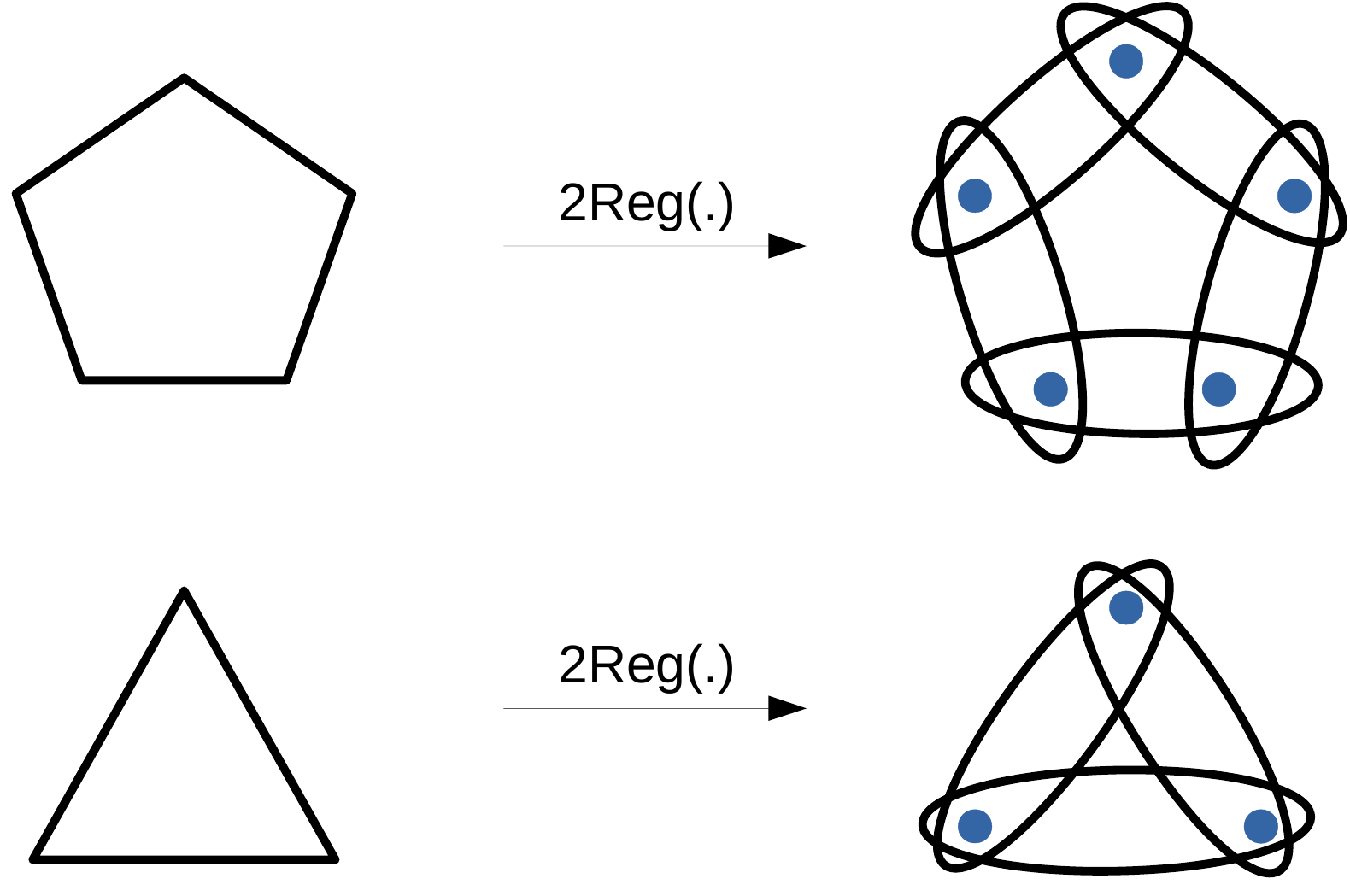}
\caption{$n$-cycles map to $n$-hypercycles and triangle maps to a 3-hypercycle under 2Reg($\cdot$).}
\label{fig3}
\end{figure}

\begin{lemma}
Any $n$-vertex complete graph $K_n$ is mapped under 2Reg($\cdot$) to a hypergraph with $^nC_2=\frac{n(n-1)}{2}$ hyperedges with $2(n-2)$
nodes per hyperedge and $^nC_2(n-2)=\frac{n(n-1)(n-2)}{2}$ nodes in all. 

\end{lemma}

\begin{proof}
$$K_n\equiv\{e_{12},\dots,e_{1n},e_{23},\dots,e_{2n},\dots,e_{n-1,n}\},$$ where $e_{ij}\equiv\{v_i,v_j\}$ denotes the edge connecting vertices
$v_i$ and $v_j$, where $i,j\in\{1,2,\dots,n\}$ and $i\neq j$. The hyperedges of 2Reg$(K_n)$ are now defined by 
$E_{ij}\equiv\{\{(e_{ij},e_{ik})|k\neq i,j\}_{k=1}^n,\{(e_{ij},e_{kj})|k\neq i,j\}_{k=1}^n\}$, where 
$i,j\in\{1,\dots,n\}$ and $i\neq j$. Hence, 2Reg$(K_n)$ is a hypergraph with $\frac{n(n-1)}{2}$ hyperedges with $2(n-2)$ nodes per hyperedge,
each node (e.g., $(e_{ij},e_{ik})$) appearing in two hyperedges (e.g., $E_{ij}$ and $E_{ik}$). The total number of nodes in the hypergraph is $\frac{n(n-1)(n-2)}{2}$.
From Theorem \ref{thmunc}, 2Reg$(K_n)$ is KS-uncolourable if and only if $^nC_2$ is odd.
\end{proof}

\subsubsection*{Complete bipartite graphs under 2Reg($\cdot$)}
Having looked at the behaviour of $n$-cycle graphs and complete graphs $K_n$ under 2Reg($\cdot$), 
let us now see how the family of complete bipartite graphs, $K_{m,n}$, can lead to KS-uncolorable scenarios 
and how this particular representation in terms of complete bipartite graphs helps us better understand the common 
ideas underlying constructions of many KS-uncolorable scenarios.
A contextuality scenario 2Reg$(K_{m,n})$ can be obtained from any complete bipartite graph $K_{m,n}$ in the following 
steps:
\begin{enumerate}
 \item $K_{m,n}$ is given by a set of vertices $V=\{v_1,v_2,\dots,v_m,t_1,t_2,\dots,t_n\}$ and edges 
 $E=\{e_{ij}\equiv\{v_i,t_j\}|i\in\{1,2,\dots,m\},j\in\{1,2,\dots,n\}\}$.
 \item Edges of $K_{m,n}$ become edges of 2Reg$(K_{m,n})$, so the contextuality scenario we construct will have $mn$ hyperedges,
 $\{f_{11},f_{12},\dots,f_{1n},\dots,f_{m1},\dots,f_{mn}\}$.
 \item Every pair of edges that share a node in $K_{m,n}$ defines a node in 2Reg$(K_{m,n})$, e.g., 
$$\{(e_{11},e_{12}),\dots,(e_{11},e_{1m}),(e_{11},e_{21}),\dots,(e_{11},e_{n1})\}$$ are all the $(m-1)+(n-1)$ nodes 
 in the hyperedge $f_{11}$ of 2Reg$(K_{m,n})$. Hence, every hyperedge of 2Reg$(K_{m,n})$ has $(m-1)+(n-1)$ nodes.
 \item Since every node appears in two hyperedges of 2Reg$(K_{m,n})$,
 2Reg$(K_{m,n})$ has $\frac{mn((m-1)+(n-1))}{2}$ nodes. 
 \item Hence: 2Reg$(K_{m,n})$ is a contextuality scenario with $\frac{mn((m-1)+(n-1))}{2}$ nodes carved up into 
 $mn$ hyperedges with $(m-1)+(n-1)$ nodes per hyperedge and each node appearing in two hyperedges.
 \end{enumerate}

\begin{lemma}
2Reg$(K_{m,n})$ is a KS-uncolourable contextuality scenario if and only if $mn$ $(>1)$ is odd.
\end{lemma}

\begin{proof}
This follows from Theorem \ref{thmunc} since $mn$ is the number of contexts in 2Reg$(K_{m,n})$.
\end{proof}

Examples of known KS-uncolourable contextuality scenarios that are of the type 2Reg$(K_{m,n})$:

\begin{enumerate}
 \item The 3-hypercycle (or ``triangle") contextuality scenario from $K_{1,3}$: 2Reg$(K_{1,3})$. This is the simplest KS-uncolourable scenario.
 It does not admit a KS set. See Fig.~\ref{fig4}
 \begin{figure}
 \centering
 \includegraphics[scale=0.4]{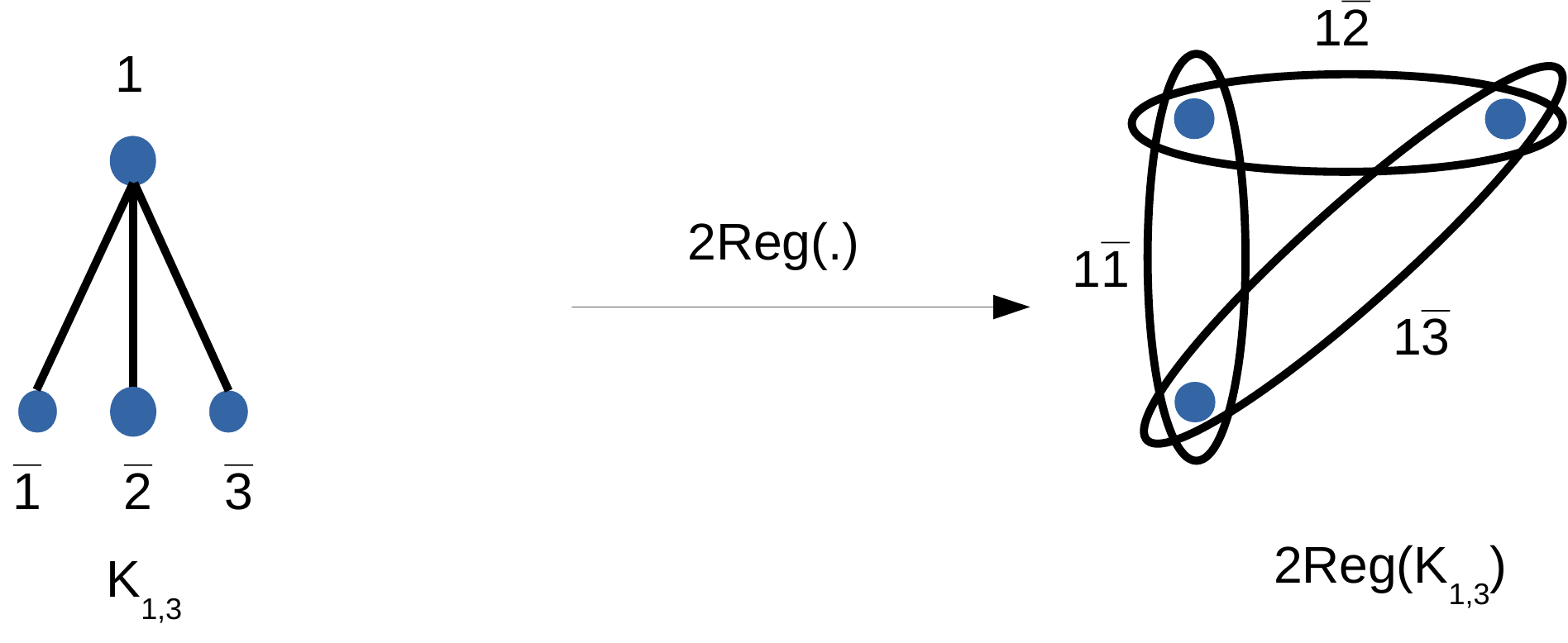}
 \caption{$K_{1,3}$ under 2Reg($\cdot$).}
 \label{fig4}
 \end{figure}

 \item 2Reg$(K_{1,5})$ in Ref.~\cite{CabelloK15}. No KS set has been found for this scenario.
 The assignments in 
 Ref.~\cite{CabelloK15} are subnormalized (that is, the projectors in a basis do not add up to identity) 
 and do not satisfy the definition of a KS set. See Fig.~\ref{fig5}.

 \begin{figure}
 \centering
 \includegraphics[scale=0.36]{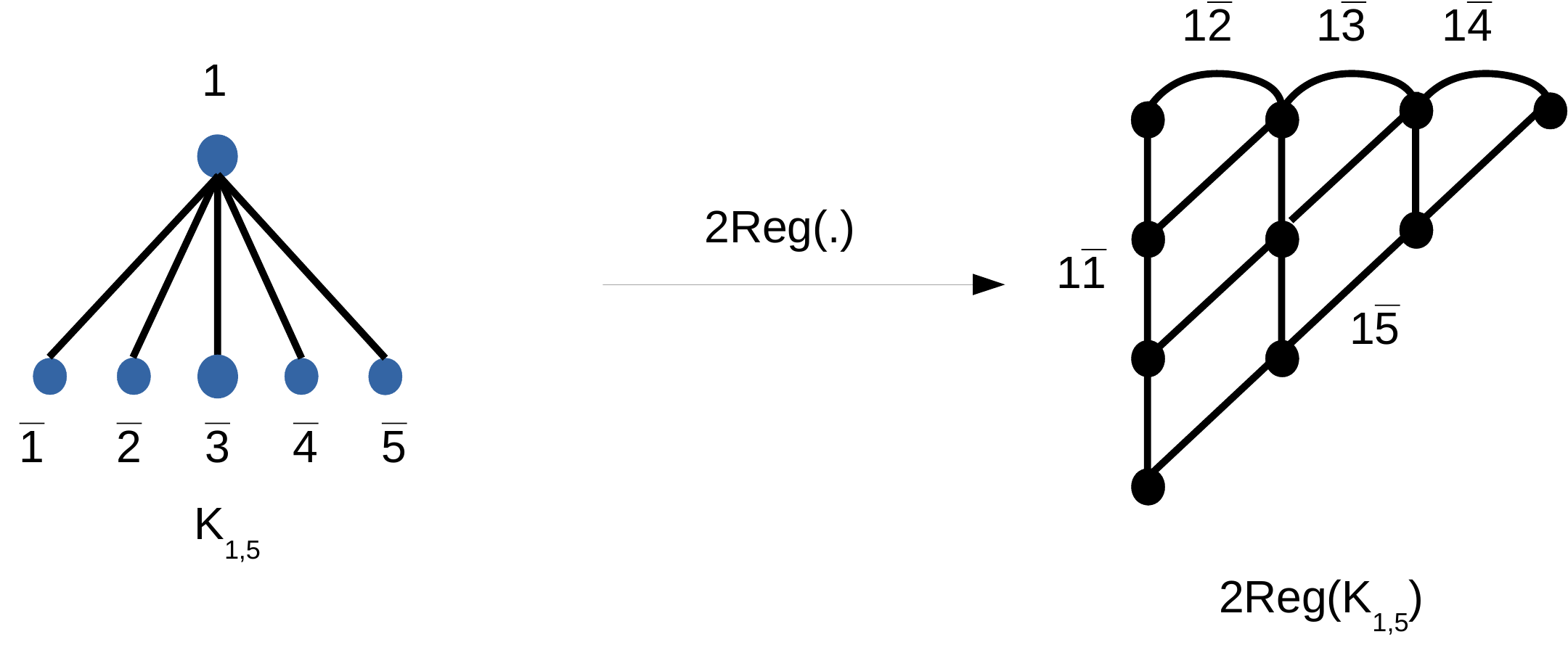}
 \caption{$K_{1,5}$ under 2Reg($\cdot$). Here we denote the hyperedges by lines rather than loops for simplicity. Each hyperedge contains 4 nodes and we have 5 hyperedges.}
 \label{fig5}
 \end{figure}

\item The seven-context KS construction (21 rays in 6 dimensions): 2Reg$(K_{1,7})$ \cite{CabelloK17}. See Fig.~\ref{fig6}.

\item The CEGA contextuality scenario (18 rays in 4 dimensions): 2Reg$(K_{3,3})$ \cite{Cabello18ray}. See Fig.~\ref{fig2}.
 \end{enumerate}

\begin{theorem}\label{2regclawtriangle}
2Reg($G$) is a $3$-hypercycle if and only if $G$ is
\begin{enumerate}
\item the claw graph, i.e., bipartite graph $K_{1,3}$, or
\item the 3-cycle or triangle graph, i.e. a graph isomorphic to $\{\{v_1,v_2\},\{v_2,v_3\},\{v_3,v_1\}\}$.
\end{enumerate}
\end{theorem}
\begin{proof}
The proof is just by explicitly exhausting all possible 3-edge graphs and verifying whether they lead to a 3-hypercycle. 

For 2Reg($G$) to be a 3-hypercycle, $G$ must have three edges, say $\{e_1,e_2,e_3\}$. Further, since each node 2Reg($G$) is defined by a pair of edges in $\{ e_1,e_2,e_3\}$, and since there are only three distinct pairs of edges in $\{ e_1,e_2,e_3\}$, it must be the case that the intersection of each pair of edges corresponds to a vertex in $G$. However, not every vertex in $G$ needs to be in the intersection of two edges. 

A vertex in $G$ can be of degree 1, 2, or 3. The fact that every pair of edges in $G$ must have a non-empty intersection means that there exists at least one vertex of degree 2 in $G$ given by $e_1\cap e_2$ or $e_2\cap e_3$ or $e_3\cap e_1$. Let us denote the edges sharing this vertex by $\{e_i,e_j\}$ ($i\neq j\in\{1,2,3\}$), so that the vertex is $e_i\cap e_j$. We thus know that $G$ has at least 3 vertices. The only option for the remaining edge, denoted $e_k$ ($k\neq i,j\in\{1,2,3\}$), in $G$ -- for 2Reg($G$) to be a 3-hypercyle -- is then to attach to the vertex  $e_i\cap e_j$, thus producing a claw graph $K_{1,3}$ with 4 vertices and 3 edges, or to attach to the two degree 1 vertices of edges $e_i, e_j$, thus forming a 3-cycle or triangle graph.
\end{proof}

\begin{corollary}\label{2regnoninvertible}
The mapping 2Reg($\cdot$) is not invertible.
\end{corollary}
\begin{proof}
This follows from Theorem \ref{2regclawtriangle}. There does not exist an inverse mapping that, when applied to 2Reg($G$) would yield $G$: in general, 2Reg($\cdot$) is a many-to-one mapping, hence non-invertible. The example of Theorem \ref{2regclawtriangle} illustrates this.
\end{proof}

\begin{theorem}\label{2regncycle}
2Reg($G$) is a $k$-hypercycle, $k\geq 4$, if and only if $G$ is a $k$-cycle.\footnote{Plus superfluous vertices of degree 0, but these can be safely ignored.}
\end{theorem}

\begin{proof}
Since 2Reg($G$) is a $k$-hypercycle, $G$ must have $k$ edges, say $\{e_1,e_2,\dots,e_k\}$. Since the nodes of 2Reg($G$) are defined by pairs of edges in $G$ with non-empty intersection, 
there must be $k$ such pairs in $\{e_1,e_2,\dots,e_k\}$. The intersection of each such pair is a vertex of $G$, hence $G$ has at least $k$ vertices. The degree of any vertex of $G$ cannot be more than 2: for any vertex of degree 3 or more in $G$ one would have the presence of a $3$-hypercycle in 2Reg($G$) (which could not then be a $k$-hypercycle, $k\geq 4$) from Theorem \ref{2regclawtriangle}. Hence, $G$ is a graph with $k$ edges and at least $k$ vertices such that $k$ pairs of edges have a non-empty intersection and each vertex of $G$ is of degree no more than 2. These constraints fix $G$ to be a $k$-cycle: any change to the $k$-cycle can only be done by adding vertices of degree 0, but no edges can be added. We can safely ignore such superfluous vertices (of degree 0) since they do not change anything about 2Reg($G$).

Combining the above with Lemma \ref{ncycle}, we have our result.
\end{proof}

The contextuality scenario 2Reg$(G)$ obtained from a graph $G$ can be viewed as a matching scenario of another graph ${\rm L}(G)$, the line graph of $G$, in the sense of Ref.~\cite{AFLS}. We comment on this connection in Appendix \ref{matchingscenarios}. In the next section we define some parameters to characterize arbitrary contextuality scenarios in a systematic manner before 
we present our general approach to obtaining noise-robust noncontextuality inequalities from KS-uncolourable scenarios.

 \begin{figure}
 \centering
 \includegraphics[scale=0.355]{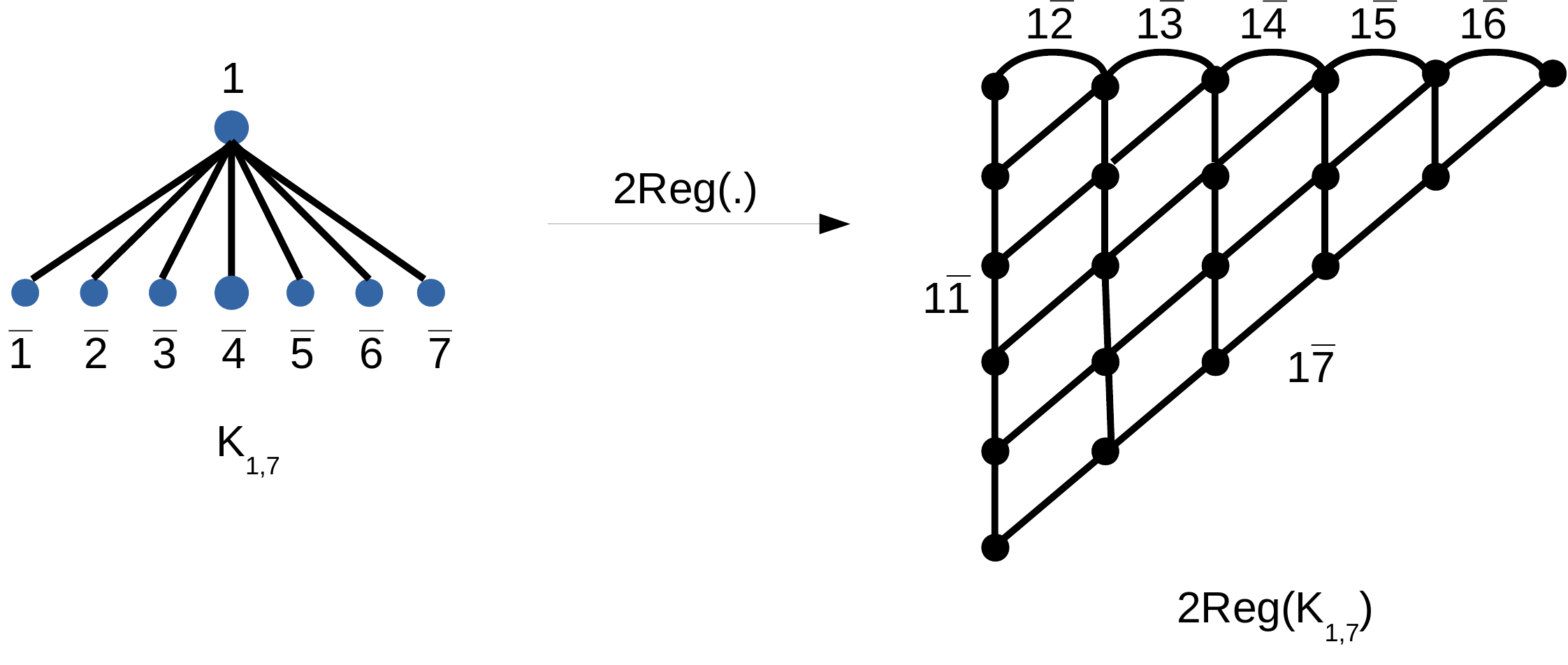}
 \caption{$K_{1,7}$ under 2Reg($\cdot$).}
 \label{fig6}
 \end{figure}

\section{Uniform contextuality scenarios, their parameters, and KS-uncolourability}

Before we can talk about noise-robust noncontextuality inequalities obtained from KS-uncolourable contextuality scenarios, we need a way to detect KS-uncolourability of a given scenario. 
We now define some parameters that we will use in our discussion of KS-uncolourability and then provide some (sufficient) conditions for KS-uncolourability of a contextuality scenario. In the process, we also answer some open questions that were posed in Ref.~\cite{MMP}. The conditions we obtain are enough for the contextuality scenarios we consider in this paper but the question of identifying conditions that are both necessary {\em and} sufficient for KS-uncolourability of a contextuality scenario remains open, despite some previous research hinting at ways in which this might be done \cite{mansthesis, caru}.

Consider a contextuality scenario, $H=(W,F)$, where $W$ is its set of nodes and $F$ its set of hyperedges. 
Further, the scenario is such that there are $d$ nodes in each hyperedge (that is, the hypergraph is $d$-uniform), $|F|$ hyperedges, and $|W|$ nodes in all.
Let $m$ be the number of nodes that appear in more than one hyperedge. We then have the following relations:
\begin{equation}
 m\leq |W| \leq d|F|.
\end{equation}
More precisely,
\begin{eqnarray}
&&d|F|=\sum_{k=1}^Dkn_k,\nonumber\\
&&|W|=\sum_{k=1}^Dn_k,\nonumber\\
&&m=\sum_{k=2}^Dn_k=|W|-n_1,
\end{eqnarray}
where $n_k$ is the number of nodes which each appear in $k$ distinct hyperedges (i.e., there exist $n_k$ nodes of degree $k$) and $D$ is the largest 
number of distinct hyperedges any node can appear in (i.e., there exists a node which appears in $D$ distinct hyperedges but no node that 
appears in $D+1$ distinct hyperedges or, equivalently, $D$ is the degree of a node of maximum degree in the hypergraph). Note that 
\begin{equation}
 \frac{d|F|}{2}=\frac{n_1}{2}+n_2+\sum_{k=3}^D\frac{kn_k}{2}.
\end{equation}
We therefore have
\begin{equation}\label{ans1}
 m\leq \frac{d|F|}{2}
\end{equation}
for any contextuality scenario. 

In Ref.~\cite{MMP}, an algorithmic way to enumerate KS-uncolourable hypergraphs which admit KS sets was presented. The observations noted
in Ref.~\cite{MMP} as a result of this algorithmic enumeration led to some conjectures, some of which we now answer by simply noting 
constraints between the parameters defined above. 

An open question in Sec.~5(i) of Ref.~\cite{MMP} was: Is it true for arbitrary $d$ that $m\leq\frac{d|F|}{2}$
(for KS-uncolourable hypergraphs which admit KS sets)? We have just answered this question in the affirmative (see Eq.~\eqref{ans1}) 
for all $d$-uniform contextuality scenarios, not just those which are KS-uncolourable and admit KS sets.
Further,
\begin{equation}
n_1=0\Leftrightarrow m=|W| \Rightarrow |W|\leq \frac{d|F|}{2}, 
\end{equation}
that is, if every node in a contextuality scenario appears in at least $2$ distinct hyperedges, then $m=|W|\leq \frac{d|F|}{2}$.
More precisely,
\begin{eqnarray}
|W|\leq \frac{d|F|}{2}\Leftrightarrow n_1&\leq& \sum_{k=3}^D(k-2)n_k\nonumber\\\nonumber
&=&n_3+2n_4+3n_5+\dots+(D-2)n_D.
\end{eqnarray}

This characterizes all the contextuality scenarios for which $|W|\leq \frac{d|F|}{2}$. 
Hence, for any $n_1>0$, we must have enough nodes with degree greater than $2$ for the relation $|W|\leq \frac{d|F|}{2}$
to hold. An open problem in Sec.~5(ii) of Ref.~\cite{MMP} asks: Given a fixed $d$, at what value of $|W|$ does the inequality 
$|W|\leq \frac{d|F|}{2}$ cease to hold? We have shown that $|W|>\frac{d|F|}{2}$ iff $n_1>\sum_{k=3}^D(k-2)n_k$. That is, roughly
speaking, the inequality $|W|\leq \frac{d|F|}{2}$ ceases to hold when there are way too many nodes of degree 1 than there are nodes of 
degree 3 or higher. 

Indeed, the only known exception to $|W|\leq\frac{d|F|}{2}$ that Ref.~\cite{MMP} finds is the construction 
due to Kochen and Specker \cite{KS67} (henceforth called the ``KS67 construction'') with 
$|W|=192$, $d=3$, $|F|=118$, so that $\frac{d|F|}{2}=177$ and we have $|W|>\frac{d|F|}{2}$. However, it is still the case that $m=117<\frac{d|F|}{2}=177$ 
(strict inequality because some of the nodes appear in more than 2 hyperedges). Note that in this case $n_1=75$ which is way greater than
$n_3+7n_9=45$.

\subsection{Conditions on the parameters of a contextuality scenario for its KS-uncolourability}

We have $d|F|=\sum_{k=1}^Dkn_k$, $|W|=\sum_{i=1}^Dn_k$, and $m=|W|-n_1$ for any $d$-uniform contextuality scenario. We want to find 
conditions that rule out the existence of a $\{0,1\}$-valued solution for 
the system of $|F|$ normalization equations in $|W|$ variables with $d$ variables in each equation. 

\subsubsection{$2$-regular contextuality scenarios}
From Theorem \ref{thmunc} we know that any $2$-regular contextuality scenario is KS-uncolourable if and only if it has an odd number of 
contexts.

\subsubsection{$d$-uniform and $2$-regular contextuality scenarios of type 2Reg$(G)$}
For $d$-uniform and 2-regular contextuality scenarios (every node of degree 2) we have $m=|W|=\frac{d|F|}{2}$. There are $|F|$ normalization 
equations in $\frac{d|F|}{2}$ variables, each variable appearing in two equations. 

For any graph $G$, 2Reg$(G)$ is a 2-regular contextuality scenario and 2Reg$(G)$ is KS-uncolourable 
if and only if $|F|$ is odd (Theorem \ref{thmunc}). If we further require 2Reg$(G)$ to be a $d$-uniform contextuality scenario,
then $d$ must be even for any 2Reg$(G)$ if $|F|$ 
is odd.\footnote{For $|W|=\frac{d|F|}{2}$ to be an integer.} 
\begin{framed}
Hence, KS-uncolourability holds for those (and only those) $d$-uniform
2Reg$(G)$ contextuality scenarios which have even $d$ and odd $|F|$: any KS set satisfying such KS-uncolourability can therefore
only be constructed on an even-dimensional Hilbert space.
\end{framed}
All proofs of the Kochen-Specker theorem relying on $d$-uniform contextuality scenarios of type 2Reg$(G)$ must therefore 
require an even-dimensional Hilbert space. In other words, it is impossible to realize KS sets for these scenarios in odd dimensions.

\subsubsection{$d$-uniform contextuality scenarios}
For more general $d$-uniform hypergraphs --- that is, not necessarily restricted to those of the type 2Reg$(G)$ --- we have: $|F|$ equations with $|W|$ variables, $n_k$ of them appearing in exactly $k$ equations (for all $k\in\{1,2,\dots,D\}$),
and each equation a sum of $d$ variables adding up to 1.
Such a hypergraph is said to be KS-uncolourable when these equations do not admit a $\{0,1\}$-valued solution. We now give some sufficient
conditions for KS-uncolourability of these hypergraphs. 

Let us denote by $W_k$ the set of nodes such that each node in the set
appears in $k$ equations (contexts), i.e., $W_k=\{w^{(k)}_j|j\in\{1,\dots,n_k\}\}$, where $k\in\{1,\dots,D\}$. 
The cardinality of the set $W_k$
is $n_k$. On adding up the $|F|$ equations, we have:
 \begin{equation}
  \sum_{k=1}^D\left(k\sum_{j=1}^{n_k}p(w^{(k)}_j)\right)=|F|,
 \end{equation}
where $p(w^{(k)}_j)\in\{0,1\}$ denotes the value assigned to node $w^{(k)}_j$. 

\begin{lemma}
$\sum_{k=1}^D\left(k\sum_{j=1}^{n_k}p(w^{(k)}_j)\right)$, where $p(w^{(k)}_j)\in\{0,1\}$, is an even number if and only if
an even number of nodes of each odd degree $k$ (such that $n_k>0$) are assigned the value 1, i.e.,
$$\sum_{j=1}^{n_k}p(w^{(k)}_j)$$ is an even number for all odd $k$ with $n_k>0$.
\end{lemma}
\begin{proof}
 This should be clear from noting that 
\begin{eqnarray}
&&\sum_{k=1}^D\left(k\sum_{j=1}^{n_k}p(w^{(k)}_j)\right)\\\nonumber
&=&\sum_{k\hspace{5pt}{\rm even}}\left(k\sum_{j=1}^{n_k}p(w^{(k)}_j)\right)+\sum_{k\hspace{5pt}{\rm odd}}\left(k\sum_{j=1}^{n_k}p(w^{(k)}_j)\right).
\end{eqnarray}
The even $k$ part of the sum is even simply because all the terms in the sum over $k$ are even. The odd $k$ part of the sum over $k$ 
is even if and only if in each term $\sum_{j=1}^{n_k}p(w^{(k)}_j)$ in the sum is even, i.e.,
$$\sum_{j=1}^{n_k}p(w^{(k)}_j)$$ is an even number for all odd $k$ with $n_k>0$.
\end{proof}

\begin{framed}
\begin{lemma}\label{parityks}
A KS contradiction arises if one of the following holds:
\begin{enumerate}
 \item $|F|$ is odd and $\sum_{j=1}^{n_k}p(w^{(k)}_j)$ is an even number for all odd $k$ with $n_k>0$.
 \item $|F|$ is even and there exists an odd $k$ with $n_k>0$ such that $\sum_{j=1}^{n_k}p(w^{(k)}_j)$ is an odd number.
\end{enumerate}
\end{lemma}
\end{framed}
\begin{proof}
 This is trivially the case because an even number $\neq$ an odd number.
\end{proof}

\subsection{The Kochen-Specker (1967) construction}
In terms of the parameters we have defined, the KS-uncolourability of the KS67 construction can be 
seen from the following:

\begin{itemize}
 \item $d=3$, $n_1=75 (=192-117)$, $n_2=90$, $n_3=24$, $n_9=3$, and $n_k=0$ for all other $k$.
 \item $|W|=n_1+n_2+n_3+n_9=192$.
 \item $m=n_2+n_3+n_9=90+24+3=117$.
 \item $d|F|=n_1+2n_2+3n_3+9n_9=75+180+72+27=354$, hence $|F|=354/3=118$.
 \item $m=|W|-n_1=117<|W|=192$, $m<d|F|/2=177$ but $|W|>d|F|/2=177$.
 \item $d=3$, $|F|=118$, $|W|=192$, $m=117$.
\end{itemize}

On adding up the 118 normalization constraints on KS67, we have:
\begin{eqnarray}
&&\sum_{j=1}^{75}p(w^{(1)}_j)+2\sum_{j=1}^{90}p(w^{(2)}_j)\nonumber\\
&+&3\sum_{j=1}^{24}p(w^{(3)}_j)+9\sum_{j=1}^3p(w^{(9)}_j)\nonumber\\
&=&118.
\end{eqnarray}
We know that one of the normalization equations is $$\sum_{j=1}^3p(w^{(9)}_j)=1,$$ hence there exists an odd $k=9$ 
(with $n_k=n_9=3>0$) such that $\sum_{j=1}^3p(w^{(9)}_j)=1,$ an odd number. Since $|F|=118$ is even, we have a 
KS contradiction for this hypergraph from Lemma \ref{parityks}.

\section{Contextuality scenarios and extremal probabilistic models on them}

Once we know that a contextuality scenario is KS-uncolourable, what can we say about the structure of extremal probabilistic models on it? If the extremal probabilistic models on a contextuality scenario admit a ``nice" characterization, then this can be used in obtaining noise-robust noncontextuality inequalities. In this section, we consider contextuality scenarios of type 2Reg($\cdot$) and extremal probabilistic models on them. Later we will use this characterization for obtaining noise-robust noncontextuality inequalities.

Note that since 2Reg$(G)$ are matching scenarios in the sense of Ref.~\cite{AFLS}, our characterization of extremal probabilistic models presented below is already known from graph-theoretic methods used in Ref.~\cite{AFLS}. All we have done below is to provide an alternative
self-contained proof that proceeds directly from Theorem \ref{extremals} instead of relying on known results from graph theory.
This is partly motivated by a need to explore in more generality (than done in Ref.~\cite{AFLS}) the consequences of Theorem \ref{extremals}, many of which we will later use in obtaining noise-robust
noncontextuality inequalities.

\begin{theorem}\label{halfinteger}
 Each extremal probabilistic model on 2Reg$(G)$ for any graph $G$ is an extension of the unique probabilistic model 
 on an induced subscenario consisting of some set of disjoint odd $k$-hypercycles, $k\in\{3,5,\dots\}$, and/or some set of 
 singleton contexts. Hence, each extremal probabilistic model on 2Reg$(G)$ assigns probabilities from $\{0,\frac{1}{2},1\}$ to the nodes 
 of 2Reg$(G)$.
\end{theorem}

\begin{proof}
We know that each extremal probabilistic model $p$ on $H\equiv$ 2Reg$(G)$ is in one-to-one correspondence with the set of nodes to which 
it assigns nonzero probability: $S_p\equiv\{w\in W(H)|p(w)\neq 0\}$.

For any graph $G$, every node in $H(=$ 2Reg$(G))$ appears in two contexts. Hence, in any induced subscenario $H_{S_p}$ defined 
by a subset of nodes of $H$ given by $S_p$ (for extremal probabilistic model $p$ on $H$), 
none of the nodes can appear in more than two contexts and none of them can be assigned probability zero.
This leaves two possibilities for nodes in $H_{S_p}$: either a node appears in one context or it appears in two contexts.

Let us denote by $q_{S_p}$ the restriction of 
extremal probabilistic model $p$ on $H$ to the unique probabilistic model $q_{S_p}$ on $H_{S_p}$: $q_{S_p}(w)=p(w)>0$ for all $w\in S_p$ (and $p(w)=0$ for all $w\in W\backslash S_p$).

Consider the set of nodes $S_p^{(1)}\subseteq S_p$ such that each of these nodes appears in exactly one context in $H_{S_p}$ and the subhypergraph $H_1$ 
obtained from $H_{S_p}$ by deleting all nodes in $S_p\backslash S_p^{(1)}$ and all hyperedges $f\in F(H_{S_p})$ such that $f\cap S_p^{(1)}=\varnothing$.
$H_1$ is then a hypergraph consisting of $|S_p^{(1)}|$ nodes, each contained in its own singleton hyperedge, hence $H_1$ admits the unique probabilistic model $q_{S_p^{(1)}}$
assigning probability 1 to each node: $q_{S_p^{(1)}}(w)=p(w)=1$ for all $w\in S_p^{(1)}$. That is, $H_1$ is a union of (disjoint) singleton contexts.

Now consider the remaining set of nodes $S_p^{(2)}\subseteq S_p$ (where $S_p^{(2)}\equiv S_p\backslash S_p^{(1)}$) such that each node appears in exactly two contexts in $H_{S_p}$ and the subhypergraph 
$H_2$ obtained from $H_{S_p}$ by deleting all nodes in $S_p^{(1)}$ and all hyperedges $f\in F(H_{S_p})$ such that $f\cap S_p^{(2)}=\varnothing$.

We will now show that $H_2$ is a union of disjoint odd $k$-hypercycles, $k\geq 3$: for $H_2$, the number of nodes of degree 2 is $n_2=|S_p^{(2)}|$,
and the restriction of extremal probabilistic model $p>0$ on $H$ to nodes in $H_2$, defined by $q_{S_p^{(2)}}(w)=p(w)$ for all $w\in S_p^{(2)}$, is also extremal on $H_2$ (otherwise $p>0$ on 
$H$ can't be extremal). Now, the existence of a probabilistic model $q_{S_p^{(2)}}>0$ on $H_2$ that is extremal $\Leftrightarrow$ $q_{S_p^{(2)}}>0$ on $H_2$ is unique (from Theorem \ref{extremals})
$\Rightarrow$ $|F(H_2)|\geq|W(H_2)|$ (where $W(H_2)=S_p^{(2)}$). 
From $n_2(H_2)=|W(H_2)|$, it follows that
$\sum_{f\in F(H_2)}d(f)=2|W(H_2)|$, where $d(f)>1$ for all $f\in F(H_2)$ (because $d(f)=1$ for any $f\in F(H_2)$
would be in conflict with the requirement that $n_2(H_2)=|W(H_2)|$ and $q_{S_p^{(2)}}>0$ on $H_2$). 
Now: $$\min_{d(f):f\in F(H_2)}\left\{\sum_{f\in F(H_2)}d(f)\right\}=2|F(H_2)|,$$
achieved at $d(f)=2$ for all $f\in F(H_2)$. Hence, $$\sum_{f\in F(H_2)}d(f)\geq 2|F(H_2)|.$$
But since $\sum_{f\in F(H_2)}d(f)=2|W(H_2)|$ and $|W(H_2)|\leq|F(H_2)|$,
we have that $$\sum_{f\in F(H_2)}d(f)\leq2|F(H_2)|.$$
Overall, $$2|F(H_2)|\leq\sum_{f\in F(H_2)}d(f)\leq2|F(H_2)|,$$
which means that $\sum_{f\in F(H_2)}d(f)=2|F(H_2)|$. 
We therefore have: $|F(H_2)|=|W(H_2)|$ and $d(f)=2$ for all $f\in F(H_2)$.
This gives us the following characterization of $H_2$:
\begin{center}
$n_2(H_2)=|W(H_2)|$, $|F(H_2)|=|W(H_2)|$, and $d(f)=2$ for all $f\in F(H_2)$\\ $\Leftrightarrow$ \\ $H_2$ is a union of disjoint $k$-hypercycles $(k\geq 3)$.
\end{center}

This follows from noting that, firstly, $H_2$ is a $2$-uniform hypergraph (that is, $d(f)=2$ for all $f\in F(H_2)$), hence really a graph.
Secondly, a $k$-cycle ($k\geq3$) is defined as a connected graph where each vertex is of degree 2 and the number of edges is equal to the number of vertices.
Hence, a graph where each vertex is of degree 2 and the number of edges is equal to the number of vertices is a union of disjoint $k$-cycles.
$H_2$ is just such a (hyper)graph.

Together with the fact that $H_2$ admits the unique probabilistic model $q_{S_p^{(2)}}>0$, this means that the $k$-hypercycles in the disjoint union have,
in fact, odd $k\geq 3$. This is because even $k$-hypercycles admit non-unique 
deterministic extremal probabilistic models which allow for assignment of probability $0$ to some nodes.
That is,

\begin{center}
$n_2(H_2)=|W(H_2)|$, $|F(H_2)|=|W(H_2)|$, $d(f)=2$ for all $f\in F(H_2)$, and $H_2$ admits a unique probabilistic model 
\\ $\Leftrightarrow$ \\ $H_2$ is a union of disjoint odd $k$-hypercycles $(k\geq 3)$.
\end{center}

This in turn implies that the unique probabilistic model $q_{S_p^{(2)}}>0$ is in fact given by $q_{S_p^{(2)}}(w)=\frac{1}{2}$
for all $w\in S_p^{(2)}$.

Hence: every extremal probabilistic model $p$ on $H =$ 2Reg$(G)$ is given by the induced subscenario $H_{S_p}$ consisting of a disjoint union of $H_1$
and $H_2$, the former a set of singleton contexts and the latter a union of disjoint odd $k$-hypercycles. Overall, $p(w)=1$ for all $w\in S_p^{(1)}$,
$p(w)=\frac{1}{2}$ for all $w\in S_p^{(2)}$, and $p(w)=0$ for all $w\in W\backslash S_p$ (where $S_p=S_p^{(1)}\sqcup S_p^{(2)}$).
 
\end{proof}

Combining Theorems \ref{2regclawtriangle}, \ref{2regncycle}, and \ref{halfinteger}, we have that we only need to look for the presence of $K_{1,3}$ and $k$-cycles ($k\geq 3$) in any graph $G$ in order to ascertain all the extremal probabilistic models on the scenario 2Reg($G$).

\begin{corollary}\label{2regextremals}
	$G$ contains a subgraph $K_{1,3}$ or an odd $k$-cycle ($k\geq 3$) if and only if the indeterministic extremal probabilistic models on 2Reg($G$) correspond to $k$-hypercycle extremal probabilistic models, $k\geq3$.
\end{corollary}

\begin{theorem}\label{khypercyclemodels} Every scenario 2Reg($K_{m,n}$) constructed from bipartite graph $K_{m,n}$ (with $mn>1$ odd) containing a bipartite subgraph $K_{1,k}$ or $K_{k,1}$, where $k (\leq m,n)$ is odd, admits extremal probabilistic model(s) corresponding to induced subscenarios obtained from a union of the $k$-hypercycle with some singletons.
\end{theorem}
\begin{proof}
Taking the complement of $K_{m,n}$, once subgraph $K_{1,k}$ or $K_{k,1}$ is removed, we denote the resulting graph as $G(K_{m,n}\backslash K_{1,k})$ or $G(K_{m,n}\backslash K_{k,1})$, respectively. Since both $G(K_{m,n}\backslash K_{1,k})$ and  $G(K_{m,n}\backslash K_{k,1})$ have $mn-k$ (an even number) of edges, from Theorem \ref{thmunc} we have that 2Reg($G(K_{m,n}\backslash K_{1,k})$) and 2Reg($G(K_{m,n}\backslash K_{k,1})$) are KS-colourable and therefore admit extremal probabilistic models induced entirely by singletons.\footnote{Note that when $k=n$, $G(K_{m,n}\backslash K_{1,k=n})$=$K_{m-1,n}$, and when $k=m$, $G(K_{m,n}\backslash K_{k=m,1})$=$K_{m,n-1}$. That is, the subgraphs are complete bipartite graphs in these cases, but not otherwise.} Now consider an induced subscenario of 2Reg($G(K_{m,n}\backslash K_{1,k})$) or 2Reg($G(K_{m,n}\backslash K_{k,1})$) that consists entirely of singletons so that the number of hyperedges in 2Reg($G(K_{m,n}\backslash K_{1,k})$) or 2Reg($G(K_{m,n}\backslash K_{k,1})$) is twice the number of singletons in this induced subscenario.\footnote{This factor of two arises because the contextuality scenarios are 2-regular, i.e., every vertex appears in two hyperedges.}
Extending this induced subscenario by adding a $k$-hypercycle (disjoint from the subscenario) leads to an induced subscenario of 2Reg($K_{m,n}$), namely, one that is a union of a $k$-hypercycle with the singletons. (See Fig.~\ref{inducedsubscenario} for an illustration in the case of $K_{3,3}$.)
\end{proof}

\section{Noise-robust noncontextuality inequalities from a hypergraph invariant}
We are finally in a position to use the understanding developed so far to obtain noise-robust noncontextuality inequalities for KS-uncolourable contextuality scenarios. The noise-robust noncontextuality inequalities reported in Ref.~\cite{KunjSpek} and the more fine-grained ones for the case of the 18 ray scenario \cite{Cabello18ray} reported in Ref.~\cite{anithesis} will be seen to be special cases of our inequalities.  We begin with an outline of the general framework within which our noise-robust noncontextuality inequalities will be obtained. In particular, we will introduce a hypergraph invariant and make precise the role that it plays in our inequalities. This framework is applicable to any KS-uncolourable contextuality scenario and we will study some well-known examples of such scenarios. This is in contrast to the framework of Ref.~\cite{robustcsw} which is only applicable to contextuality scenarios that are KS-colourable and also satisfy the property that all probabilistic models on them obey consistent exclusivity \`a la AFLS \cite{AFLS}.

\begin{figure}[htb!]\centering
	\includegraphics[scale=0.31]{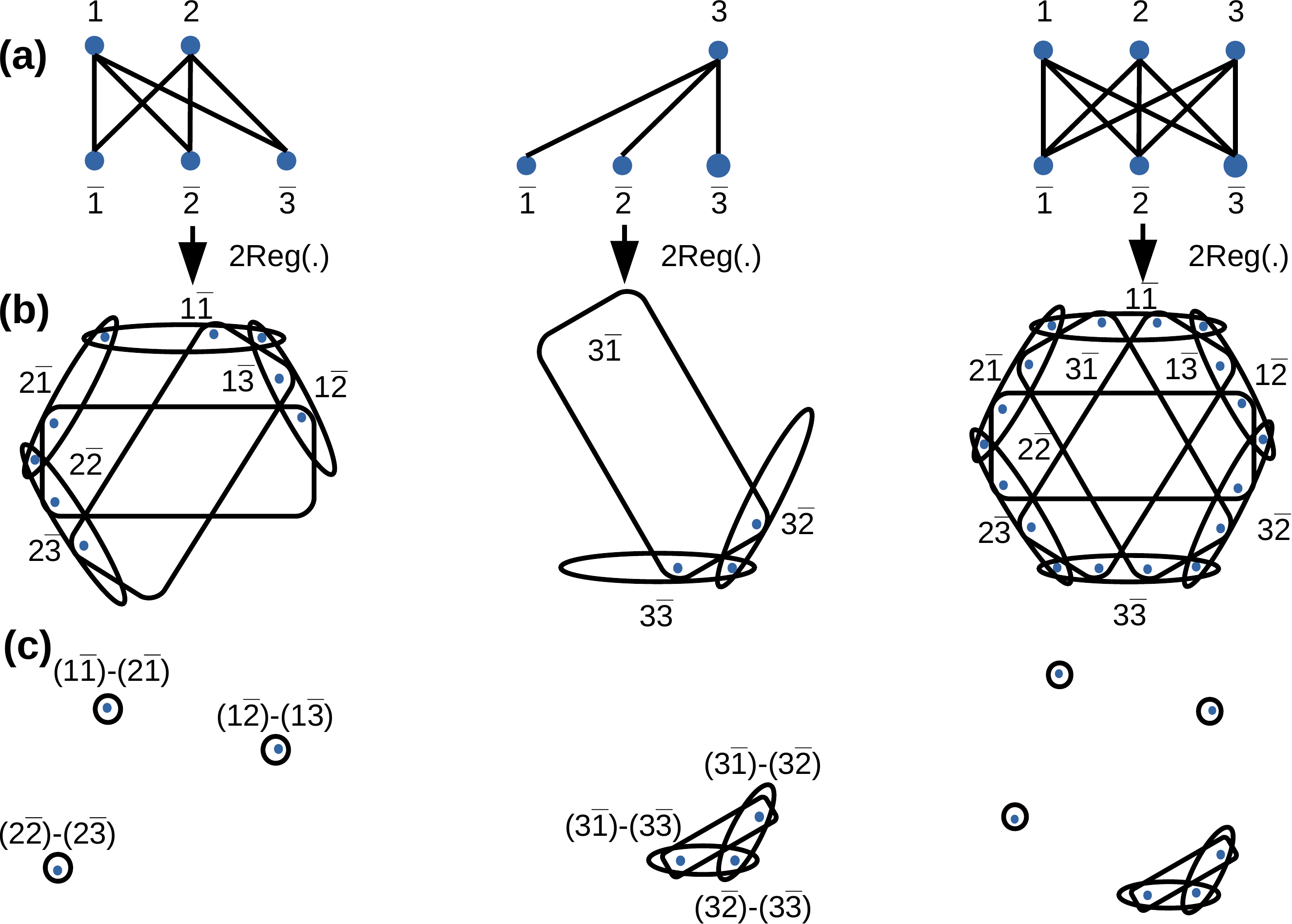}
	\caption{(a) $K_{3,3}$ and its two subgraphs, $K_{2,3}$ and $K_{1,3}$, (b) the contextuality scenarios obtained under the mapping 2Reg($\cdot$), and (c) the induced subscenarios corresponding to extremal probabilistic models on the respective contextuality scenarios, 2Reg($K_{2,3}$), 2Reg($K_{1,3}$), and 2Reg($K_{3,3}$). The induced subscenario corresponding to an extremal probabilistic model on 2Reg($K_{3,3}$) is a disjoint union of induced subscenarios corresponding to 2Reg($K_{2,3}$) and  2Reg($K_{1,3}$).}
	\label{inducedsubscenario}
\end{figure}
\twocolumn
\subsection{Operational equivalences}
In keeping with the treatment in Ref.~\cite{KunjSpek}, we associate with each KS-uncolourable scenario two kinds of hypergraphs: one corresponding to the operational equivalences presumed between measurement events and another corresponding to operational equivalences presumed between source events.\footnote{We are assuming that in any experimental test of noncontextuality, these operational equivalences have been established \cite{exptlpaper}.}

Assuming the contextuality scenario consists of $n$ (measurement) contexts with $d$ nodes each, we consider $n$ measurement settings $M_i$, $i\in\{1,2,\dots,n\}\equiv[n]$, each with $d$ possible outcomes, $m_i\in\{1,2,\dots,d\}\equiv[d]$. We assume operational equivalences between the measurement outcomes $[m_i|M_i]$ that are reflected in the contextuality scenario. These are of the type: $[m_i|M_i]\simeq[m_j|M_j]$, for some pairs $\{i,j\}\subset\{1,2,\dots,d\}$. The assumption of measurement noncontextuality then says: $\xi(m_i|M_i,\lambda)=\xi(m_j|M_j,\lambda)\quad\forall \lambda\in\Lambda$.

We also consider source settings $S_i$, $i\in[n]$, each with $d$ possible outcomes $s_i\in[d]$, 
such that the following operational equivalences hold among the sources:
\begin{align}
\forall [m|M]:&\sum_{s_i=1}^dp(m,s_i|M,S_i)=\sum_{s_j=1}^dp(m,s_j|M,S_j),\nonumber\\
&\textrm{ for all }i,j\in[n].
\end{align}
That is, 
\begin{equation}
[\top|S_1]\simeq [\top|S_2]\simeq\dots\simeq [\top|S_n].
\end{equation}

Given that $$p(m,s|M,S)=\sum_{\lambda\in\Lambda}\xi(m|M,\lambda)\mu(s|S,\lambda)\mu(\lambda|S),$$ the assumption of 
preparation noncontextuality then says:
\begin{equation}
\mu(\lambda|S_1)=\mu(\lambda|S_2)=\dots=\mu(\lambda|S_n)\equiv\nu(\lambda)\quad\forall\lambda\in\Lambda.
\end{equation}
Here, 
\begin{equation}
 \mu(\lambda|S_i)\equiv\sum_{s_i=1}^d\mu(\lambda|S_i,s_i)p(s_i|S_i),\textrm{for all } i\in[n].
\end{equation}

We recall the measurement events and source events hypergraphs for the example of Ref.~\cite{KunjSpek} in Fig.~\ref{fig7}, where $n=9$ and $d=4$. 
 \begin{figure}
 \centering
 \includegraphics[scale=0.4]{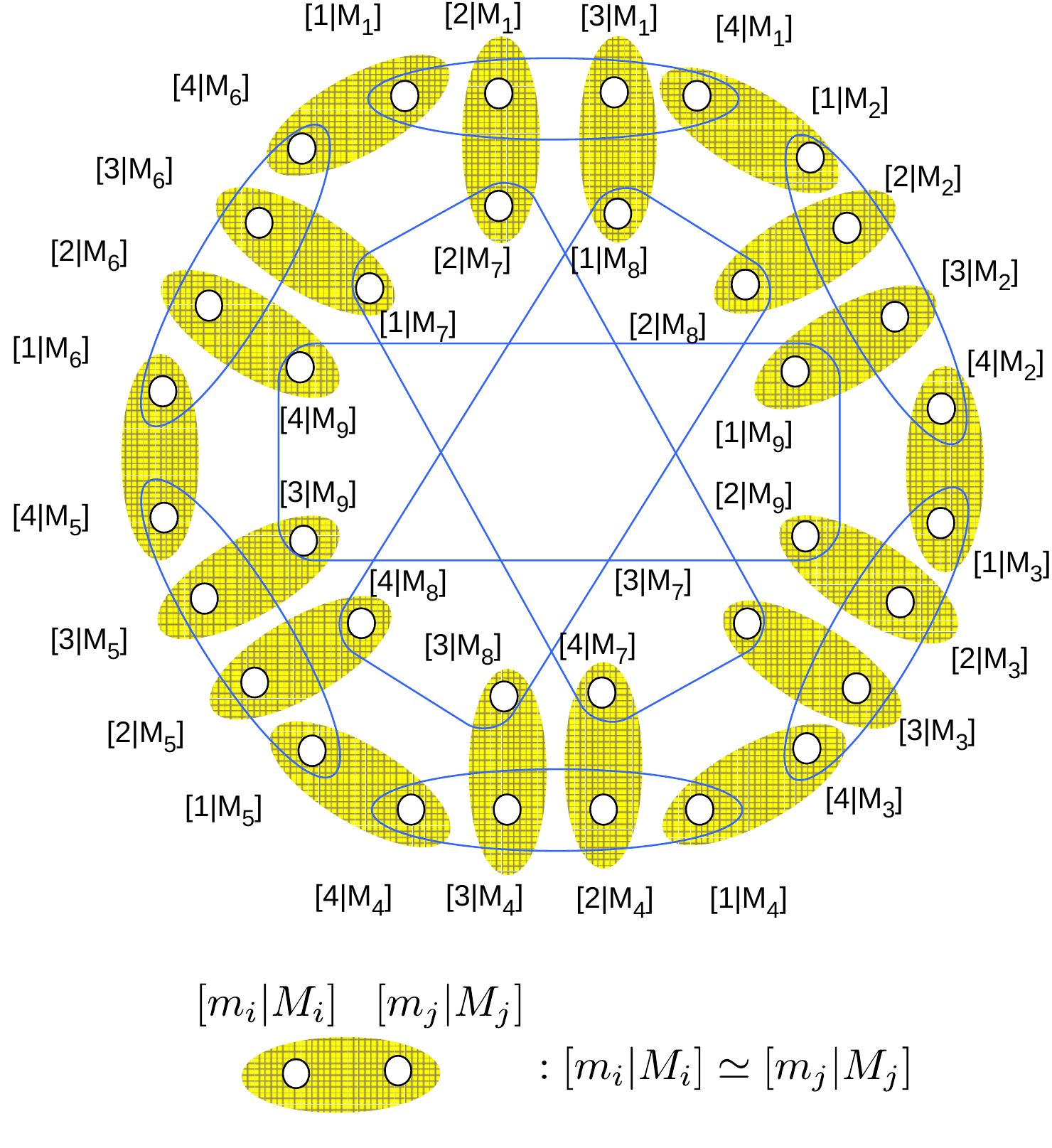}
 \includegraphics[scale=0.4]{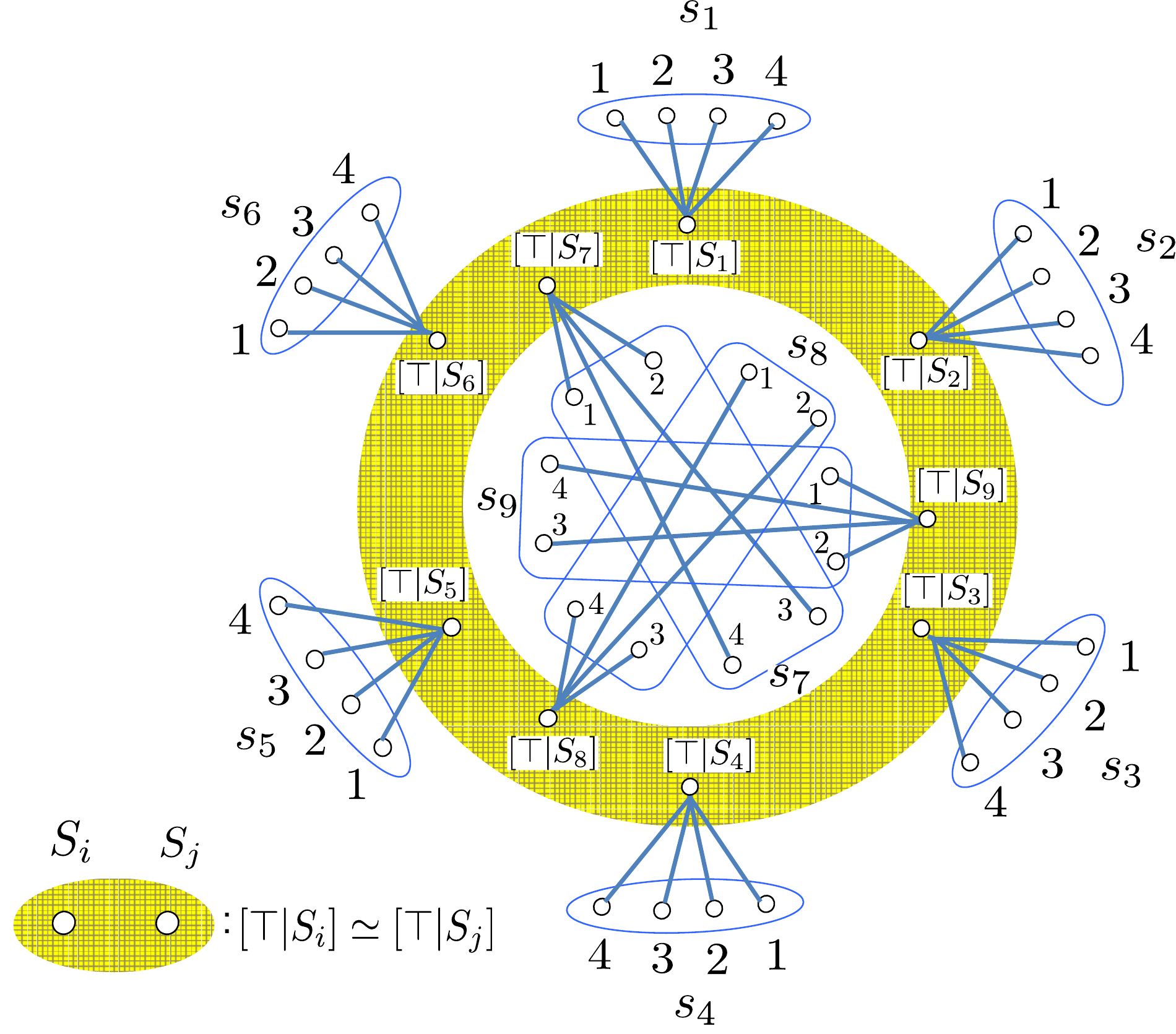}
 \caption{Operational equivalences between measurements (top) and operational equivalences between sources (bottom).}
 \label{fig7}
 \end{figure}
The noncontextuality inequality of Ref.~\cite{KunjSpek} that follows from the operational equivalences for sources and measurements then reads:
\begin{align}
 A&\equiv\frac{1}{9}\sum_{i=1}^9\frac{1}{4}\sum_{m_i=1}^4p(m_i|M_i,S_i,s_i=m_i)\nonumber\\
 &\leq\max_{\lambda\in\Lambda}\frac{1}{9}\sum_{i=1}^9\zeta(M_i,\lambda)=\frac{5}{6},
\end{align}
where $\zeta(M_i,\lambda)\equiv\max_{m_i}\xi(m_i|M_i,\lambda)$ for all $i\in\{1,2,\dots,9\}$, so that $\frac{1}{9}\sum_{i=1}^9\zeta(M_i,\lambda)$ is the average max-probability for a 
given $\lambda\in\Lambda$. Hence,
the noncontextuality inequality bounds $A$ by the maximum average max-probability of the measurements $M_i$ possible in {\em any} ontological model. Noting that $p(s_i=m_i|S_i)=\frac{1}{4}$ for all $m_i\in[4]$ for the scenario of Ref.~\cite{KunjSpek}, we can rewrite the quantity $A$ as:
\begin{align}
A&=\frac{1}{9}\sum_{i=1}^9\sum_{s_i=1}^4p(m_i=s_i,s_i|M_i,S_i)\nonumber\\
&\equiv\frac{1}{9}\sum_{i=1}^9\sum_{x=1}^4p(m_i=x,s_i=x|M_i,S_i).
\end{align} 
In the general case of $n$ measurement procedures with $d$ outcomes each, the expression for $A$ reads
\begin{equation}
 A\equiv\frac{1}{n}\sum_{i=1}^n\sum_{x=1}^dp(m_i=x,s_i=x|M_i,S_i).
\end{equation}
Furthermore, we need not even restrict ourselves to a uniform average of the source-measurement correlation over the source and measurement settings and allow instead a weighted average given by some probability distribution $q\equiv\{q_i\}_{i=1}^{n}$, where $q_i\geq0$ for all $i$ and $\sum_iq_i=1$. The quantity $A$ then becomes the source-measurement correlation quantity ${\rm Corr}$ that was previously defined in Ref.~\cite{robustcsw} and can be upper bounded as follows (following Ref.~\cite{KunjSpek}):
\begin{align}\label{ncineqschematic}
&{\rm Corr}\nonumber\\
&=\sum_{i=1}^nq_i\sum_{x=1}^dp(m_i=x,s_i=x|M_i,S_i),\nonumber\\
&=\sum_{i=1}^nq_i\sum_{x=1}^d\sum_{\lambda\in\Lambda}\xi(m_i=x|M_i,\lambda)\mu(s_i=x|S_i,\lambda)\mu(\lambda|S_i),\nonumber\\
&\leq\sum_{i=1}^nq_i\sum_{s_i=1}^d\sum_{\lambda\in\Lambda}\max_{m_i\in[d]}\xi(m_i|M_i,\lambda)\mu(s_i|S_i,\lambda)\mu(\lambda|S_i),\nonumber\\
&=\sum_{\lambda\in\Lambda}\sum_{i=1}^nq_i\zeta(M_i,\lambda)\sum_{s_i=1}^d\mu(s_i|S_i,\lambda)\mu(\lambda|S_i),\nonumber\\
&=\sum_{\lambda\in\Lambda}\left(\sum_{i=1}^nq_i\zeta(M_i,\lambda)\mu(\lambda|S_i)\right),\nonumber\\
&=\sum_{\lambda\in\Lambda}\left(\sum_{i=1}^nq_i\zeta(M_i,\lambda)\right)\nu(\lambda),\nonumber\\
&(\textrm{using preparation noncontextuality})\nonumber\\
&\leq\max_{\lambda\in\Lambda}\sum_{i=1}^nq_i\zeta(M_i,\lambda)\nonumber\\
&(\textrm{using convexity})\nonumber\\
&\equiv\beta(\Gamma,q)<1 \textrm{ (for some choices of } q) \nonumber\\
&(\textrm{using measurement noncontextuality and}\nonumber\\
&\textrm{KS-uncolourability}).
\end{align}
Here, $\beta(\Gamma,q)$ is the {\em weighted max-predictability} for a contextuality scenario $\Gamma$, first defined in Ref.~\cite{robustcsw} as
\begin{equation}
\beta(\Gamma, q)\equiv\max_{\lambda\in\Lambda_{\rm ind}}\sum_{i=1}^nq_i\zeta(M_i,\lambda),
\end{equation}
where $\Lambda_{\rm ind}$ is the set of ontic states which all assign indeterministic extremal probabilistic models on $\Gamma$, i.e., extremal probabilistic models with probabilities valued in the interval $(0,1)$ for {\em at least} one measurement context in the contextuality scenario.
In this paper, the contextuality scenario $\Gamma$ is KS-uncolourable, hence the qualifier that the maximization in the definition of $\beta(\Gamma,q)$ is taken over only the ontic states that assign indeterministic extremal probabilistic models to measurement events in $\Gamma$ (that is, $\Lambda_{\rm ind}$) is unnecessary: {\em all} ontic states assigning probabilities to the measurement events of a KS-uncolourable $\Gamma$ must necessarily correspond to indeterministic extremal probabilistic models (i.e., $\Lambda=\Lambda_{\rm ind}$). Note also that $\beta(\Gamma,q)$ need not always be strictly less than 1 for KS-uncolourable $\Gamma$: this can happen, for example, if $q$ is supported only on those contexts (if they exist) which are all assigned $\{0,1\}$-valued probabilities (i.e., they are deterministic) by {\em some} extremal probabilistic model on $\Gamma$. On the other hand, we know that there {\em always} exists a choice of $q$ such that $\beta(\Gamma,q)<1$ simply because of the KS-uncolourability of $\Gamma$, e.g., any choice where $q$ is supported on {\em all} the contexts of $\Gamma$ --- such as $q$ being a uniform probability distribution over the measurement contexts --- so that there is no extremal probabilistic model that makes all the contexts deterministic.

In the following sections, we will show how bounds on ${\rm Corr}$ when $q$ is supported on certain (sub)sets of contexts (which we will call minimally indeterministic sets of contexts, or MISCs, below) can be obtained from conceptual arguments instead of a brute-force computational approach. We will show that 
for such MISCs (instead of {\em all} the contexts in a contextuality scenario), we can obtain noncontextuality inequalities of 
the type:

\begin{align}
{\rm Corr}_q&\equiv\sum_{i=1}^c q_{r_i}\sum_{x=1}^dp(m_{r_i}=x,s_{r_i}=x|M_{r_i},S_{r_i})\nonumber\\
&\leq\beta(\Gamma,q)<1,
\end{align}
where $c$ $(<n)$ is the number of contexts in a MISC, each context denoting a measurement setting $M_{r_i}$, where $q_{r_i}>0$ for all $i\in\{1,2,\dots,c\}$, $\sum_{i=1}^cq_{r_i}=1$, and $r_i\in\{1,2,\dots,n\}$ are all distinct.\footnote{Note that this inequality follows from the assumptions of preparation and measurement noncontextuality, as in Eq.~\eqref{ncineqschematic}, hence it is a noise-robust ``noncontextuality inequality". Unlike the noise-robust noncontextuality inequalities of Ref.~\cite{robustcsw}, however, the noncontextuality inequalities in this paper are {\em not} (in any sense) ``generalizations" of KS-noncontextuality inequalities \`a la CSW \cite{CSW}. The set of probabilistic models on a KS-uncolourable contextuality scenario that satisfy KS-noncontextuality is empty, i.e., the set of ``classical models" in the terminology of Ref.~\cite{AFLS} is empy for such scenarios. Hence, any probabilistic model on such a scenario is KS-contextual and no meaningful constraint from  KS-noncontextuality on probabilistic models can be written down. This is also reflected in the fact that our inequalities here do not invoke any of the traditional graph invariants used in Ref.~\cite{CSW}.}

If the contextuality scenario $\Gamma$ were KS-colourable, then we would have  $\max_{\lambda\in\Lambda}\sum_{i=1}^nq_i\zeta(M_i,\lambda)=1$ for any choice of $q$ (corresponding to any set of $c$ contexts); however, since $\Gamma$ is KS-uncolourable, there necessarily exist one or more sets of $c$ contexts (for some $c$) such that $\max_{\lambda\in\Lambda}\sum_{i=1}^nq_i\zeta(M_i,\lambda)<1$ when $q$ is supported on such sets. Note that while in the former case $\beta(\Gamma,q)$ is undefined, in the latter case we have $\beta(\Gamma,q)=\max_{\lambda\in\Lambda}\sum_{i=1}^nq_i\zeta(M_i,\lambda)$. The MISCs we define below are examples of such sets of $c$ contexts for which $\beta(\Gamma,q)<1$. Finding a MISC and computing its $\beta(\Gamma,q)$ value yields a noise-robust noncontextuality inequality in our framework. 

We will be interested in finding all the {\em irreducible} MISCs (or as we define them later, ``irrMISCs") in a KS-uncolourable contextuality scenario: finding them and evaluating their $\beta(\Gamma,q)$ values amounts to identifying a minimal set of independent noncontextuality inequalities for that scenario; from these, all the other MISC inequalities can be obtained by coarse-graining.

\subsection{Minimally Indeterministic Sets of Contexts (MISCs)}

We now consider assignments of probabilistic models to a contextuality scenario specified by an ontic state $\lambda\in\Lambda$
according to the response functions $\xi(m_i|M_i,\lambda)\in[0,1]$. 
A {\em deterministic context} is one where all the measurement outcomes are assigned $\{0,1\}$-valued probabilities, i.e., 
$\xi(m_i|M_i,\lambda)\in\{0,1\}$ for all $m_i$, $M_i$. An {\em indeterministic
context} is one which is not deterministic, i.e., it only allows probability assignments in $[0,1)$ to the measurement outcomes.
The max-probability for a deterministic context is $1$ while for an indeterministic context it is less than $1$.

\begin{framed}
{\bf Minimally Indeterministic Set of Contexts (MISC) of size $c$:}
A set of $c$ contexts such that {\it no more than} $c-1$ of them can be made deterministic by {\it any} (extremal) probabilistic model
on the (parent) contextuality scenario, i.e., $\beta(\Gamma,q)<1$ when $q$ is supported entirely on such a set of $c$ contexts. 
\end{framed}

Intuitively, a MISC is a subset of contexts that does not admit a KS-noncontextual assignment of outcomes arising from a restriction of any probabilistic model on the parent contextuality scenario (of which the MISC is a subset) to just the MISC. By a ``restriction" of a probabilistic model to a MISC, we mean the set of probabilities assigned to measurement events in a MISC by the probabilistic model on the parent contextuality scenario. Any proper subset of a MISC {\em might}, however, admit a KS-noncontextual assignment of outcomes arising in this way.\footnote{Mansfield and Barbosa \cite{mansbosa} have previously considered a notion of (partial) extendability of an empirical model (probabilistic model, in our terminology) on a measurement cover (i.e., a joint measurability structure \cite{KHF} which yields a contextuality scenario \cite{gondaetal}) to an empirical model on another measurement cover such that the down-closure of the first measurement cover is contained in the down-closure of the second. While this notion bears a superficial similarity to the relation between a proper subset of a MISC and the MISC itself, as we have described it, it may be an interesting avenue for further research to look into rigorously formalizing the counterpart of MISCs in the sheaf-theoretic approach \cite{AB}, possibly via Ref.~\cite{mansbosa}.}

\subsubsection{Noise-robust noncontextuality inequalities for any KS-uncolourable contextuality scenario} 
Simple noncontextuality inequalities can be obtained from a KS-uncolourable contextuality scenario by identifying the following type of MISCs:
\begin{framed}
For a KS-uncolourable contextuality scenario (with, say, $n$ contexts), every extremal probabilistic model will make some of the contexts
indeterministic. Let $k$ be the smallest number of such indeterministic contexts present in {\em any} extremal probabilistic model on the 
KS-uncolourable contextuality scenario.
Then any set of $n-k+1$ contexts (out of all the $n$) constitutes a MISC, i.e., $\beta(\Gamma,q)<1$ when $q$ is supported entirely over this set of contexts.
\end{framed}

From Theorem \ref{extremals}, for a KS-uncolourable contextuality scenario, every extremal probabilistic model is in one-to-one correspondence with
an induced subscenario admitting a unique probabilistic model. KS-uncolourability means that any induced subscenario with a unique probabilistic 
model would necessarily contain hyperedges that are non-singleton (i.e., containing more than one node) with their nodes assigned probabilities less than 1.
Ignoring the singleton hyperedges in such an induced subscenario (i.e., those containing exactly one node), all the remaining hyperedges are indeterministic.
We refer to the subscenario consisting of these remaining (indeterministic) hyperedges and the nodes they contain as an {\em induced 
indeterministic subscenario}. $k$ is then the number of contexts in the smallest (in terms of the number of contexts)
induced indeterministic subscenario obtained from an induced subscenario with a unique probabilistic model. 
Now note that the hypergraph with the least number of contexts (and containing no singleton contexts)
admitting a unique probabilistic model is a $3$-hypercycle. Hence, a $3$-hypercycle is the smallest induced indeterministic
subscenario possible and we have
\begin{framed}
{\bf Sufficient condition for a set of contexts to be a MISC: } For {\it any} KS-uncolourable contextuality scenario it will be the case that $k\geq 3$ and {\em any} set of $n-2$ contexts 
in the scenario will form a MISC, i.e., $\beta(\Gamma,q)<1$ when $q$ is supported on {\em any} set of $n-2$ contexts.
\end{framed}

For an example, see Fig.~\ref{inducedsubscenario}(c), second column, for an induced indeterministic subscenario of the 18 ray hypergraph \cite{Cabello18ray} and the third column for the induced subscenario of which the induced indeterministic subscenario is a part.

Given that $k$ is the size of the smallest induced indeterministic subscenario, we have a noncontextuality inequality whenever $q$ is supported on any set of $n-k+1$ contexts (which constitute a MISC):

\begin{eqnarray}
&&{\rm Corr}_q\nonumber\\
&\equiv&\sum_{i=1}^{n-k+1}q_{r_i}\sum_{x=1}^dp(m_{r_i}=x,s_{r_i}=x|M_{r_i},S_{r_i})\nonumber\\
&\leq& \max_{\lambda\in\Lambda}\sum_{i=1}^{n-k+1}q_{r_i}\zeta(M_{r_i},\lambda)\nonumber\\
&\equiv&\beta(\Gamma,q).
\end{eqnarray}
If we take $q_{r_i}=\frac{1}{n-k+1}$ for all $i\in\{1,2,\dots,n-k+1\}$, we have the following noncontextuality inequality for a MISC consisting of $n-k+1$ contexts:
\begin{eqnarray}
&&{\rm Corr}_q\nonumber\\
&\equiv&\frac{1}{n-k+1}\sum_{i=1}^{n-k+1}\sum_{x=1}^dp(m_{r_i}=x,s_{r_i}=x|M_{r_i},S_{r_i})\nonumber\\
&\leq& \max_{\lambda\in\Lambda}\frac{1}{n-k+1}\sum_{i=1}^{n-k+1}\zeta(M_{r_i},\lambda) \equiv \beta(\Gamma,q)\nonumber\\
&\leq&\frac{n-k}{n-k+1}+\frac{p_{\rm max}}{(n-k+1)}\nonumber\\
&=&1-\frac{1-p_{\rm max}}{n-k+1},
\end{eqnarray}
where $p_{\rm max}\in\big[\frac{1}{d},1\big]$ is the largest max-probability associated with any indeterministic context included in the MISC.
This max-probability corresponds to an extremal probabilistic model that makes all but one of the contexts in the MISC deterministic. In the case of the 18 ray scenario, for example, $k=3$ and any set of $n-k+1=9-3+1=7$ contexts forms a MISC and we have $p_{\max}=\frac{1}{2}$, so that the upper bound in the above inequality given by $\frac{13}{14}$. (See Fig.~\ref{inducedsubscenario}, third column: the six deterministic contexts together with any one of the three indeterministic contexts form such a seven-context MISC.)

\subsubsection{Sufficient condition for a set of contexts to be a MISC is not necessary}
While the sufficient condition outlined above for a set of contexts in a contextuality scenario to be a MISC works for any KS-uncolourable contextuality scenario and yields
noncontextuality inequalities, it is not a necessary condition. It is possible to identify smaller MISCs depending on the 
particular contextuality scenario and the probabilistic models on it. 

Consider, for example, all 
the scenarios of the type 2Reg$(G)$ that we have discussed. In these scenarios, each node appears in two contexts
and therefore deterministic contexts appear in pairs in any extremal probabilistic model on these scenarios: this is because 
the deterministic contexts in any extremal probabilistic model are determined by singleton hyperedges in the induced subscenario and the node 
in a singleton hyperedge (assigned probability 1) appears in two contexts in the full contextuality scenario (see Fig.~\ref{inducedsubscenario}, third column, for example). 
It then becomes possible 
to reduce the MISCs of size $n-k+1$ that we have identified above to MISCs of size $\frac{n-k}{2}+1$ simply by 
taking the given MISC and omitting one of each pair of deterministic contexts that share a node in the given MISC. For example, see Fig.~\ref{inducedsubscenario}(c), third column, where three deterministic contexts together with an indeterministic context form a four-context MISC.

Since a MISC may thus contain smaller MISCs, we define the notion of an ``irreducible MISC'':
\begin{framed}
{\bf Irreducible MISC (irrMISC):} A MISC which does not contain another MISC as a proper subset. 
\end{framed}

Therefore, an irrMISC is such that for its every proper subset there exists an extremal probabilistic model in which this proper subset is 
deterministic. As we noted, the MISCs of size $n-k+1$ we have identified above can be reduced to MISCs of size $\frac{n-k}{2}+1$ in 2Reg($\cdot$) scenarios.
Are these MISCs of size $\frac{n-k}{2}+1$ irreducible? Not necessarily. 

In general, it is possible to identify proper subsets of 
MISCs which are irreducible MISCs. Let us see how this plays out for some scenarios we will consider in detail here: 
2Reg$(K_{3,3})$, 2Reg$(K_{1,7})$, and the general case of  2Reg$(K_{1,n})$ (odd $n>1$). For concreteness, we will assume in the following subsections that $q$ is a uniform distribution over all the contexts in a MISC, although our identification of MISCs does not rely on this choice. 

After illustrating the underlying ideas via these explicit examples, we will conclude with a general theorem characterizing irrMISCs in contextuality scenarios of type $2{\rm Reg}(K_{m,n})$ (with odd $mn>1$).

\subsubsection{2Reg$(K_{3,3})$}
Denoting the edges of $K_{3,3}$ (and the corresponding hyperedges of 2Reg$(K_{3,3})$) by $$\{(1\bar{1}),(1\bar{2}),(1\bar{3}),(2\bar{1}),(2\bar{2}),
(2\bar{3}),(3\bar{1}),(3\bar{2}),(3\bar{3})\},$$ 
we can identify the following six 3-hypercycles in 2Reg$(K_{3,3})$ (See Fig.~\ref{3hypercycles}):
 \begin{eqnarray}
  &&(1\bar{1})-(1\bar{2})-(1\bar{3})-(1\bar{1})\nonumber\\
  &&(2\bar{1})-(2\bar{2})-(2\bar{3})-(2\bar{1})\nonumber\\
  &&(3\bar{1})-(3\bar{2})-(3\bar{3})-(3\bar{1})\nonumber\\
  &&(1\bar{1})-(2\bar{1})-(3\bar{1})-(1\bar{1})\nonumber\\
  &&(1\bar{2})-(2\bar{2})-(3\bar{2})-(1\bar{2})\nonumber\\
  &&(1\bar{3})-(2\bar{3})-(3\bar{3})-(1\bar{3}).
 \end{eqnarray}

It is easy to show that each of these 3-hypercycles forms a part of multiple induced subscenarios corresponding to extremal probabilistic models.
For example, see Fig.~\ref{3hypercycles} for the induced subscenarios (and corresponding extremal probabilistic models) where the 3-hypercycle 
$(3\bar{1})-(3\bar{2})-(3\bar{3})-(3\bar{1})$ appears.
Since a 3-hypercycle is the smallest hypergraph with a unique probabilistic model that isn't deterministic, we 
can build a 7-context MISC by taking any one of the edges in $(3\bar{1})-(3\bar{2})-(3\bar{3})-(3\bar{1})$ and the remaining six edges 
(out of nine). This gives three distinct MISCs for the 3-hypercycle $(3\bar{1})-(3\bar{2})-(3\bar{3})-(3\bar{1})$:

\begin{eqnarray}
&&{\rm MISC_1(7)}\equiv\{(3\bar{1}),(2\bar{1}),(2\bar{2}),(2\bar{3}),(1\bar{1}),(1\bar{2}),(1\bar{3})\},\nonumber\\
&&{\rm MISC_2(7)}\equiv\{(3\bar{2}),(2\bar{1}),(2\bar{2}),(2\bar{3}),(1\bar{1}),(1\bar{2}),(1\bar{3})\},\nonumber\\
&&{\rm MISC_3(7)}\equiv\{(3\bar{3}),(2\bar{1}),(2\bar{2}),(2\bar{3}),(1\bar{1}),(1\bar{2}),(1\bar{3})\}.\nonumber\\
\end{eqnarray}
The noncontextuality inequality corresponding to each 7-context MISC (${\rm MISC_j(7)},j=1,2,3$) is given by

\begin{eqnarray}
&&{\rm Corr}_{\rm MISC_j(7)}\nonumber\\
&\equiv&\frac{1}{7}\sum_{i\in {\rm MISC_j(7)}}\sum_{x=1}^4p(m_i=x,s_i=x|M_i,S_i)\nonumber\\
&\leq&\max_{\lambda\in\Lambda}\frac{1}{7}\sum_{i\in {\rm MISC_j(7)}}\zeta(M_i,\lambda)\equiv \beta(\Gamma,q)\nonumber\\
&=&\frac{1}{7}\left(6+\frac{1}{2}\right)=\frac{13}{14}.
\end{eqnarray}

More generally, from the fact that the contextuality scenario admits 3-hypercycles, we have that the average predictability is 
constrained for any set of $c\geq7$ contexts (out of $9$) since {\em at most} $6$ of them can be made deterministic
but not the remaining ones by any extremal probabilistic model. Then for any choice of $c$ contexts such that $c=7,8,9$, we have noncontextuality 
inequalities with ${\rm Corr}_{{\rm MISC}_j(c)}$constrained by $13/14$, $7/8$, and $5/6$ respectively.

The $7$-context MISCs can be further reduced to $4$-context MISCs by eliminating one of each pair of contexts from the remaining 
deterministic edges, $\{(2\bar{1}),(2\bar{2}),(2\bar{3}),(1\bar{1}),(1\bar{2}),(1\bar{3})\}$, in the MISC. 
For each ${\rm MISC_j(7)}$, for instance, these pairs of 
deterministic contexts can be: 
\begin{eqnarray}
&&\{(2\bar{1})-(2\bar{2}),(2\bar{3})-(1\bar{3}),(1\bar{1})-(1\bar{2})\},\nonumber\\
&&\{(2\bar{1})-(1\bar{1}),(2\bar{3})-(1\bar{3}),(2\bar{2})-(1\bar{2})\},\nonumber\\
&&\{(2\bar{1})-(1\bar{1}),(2\bar{3})-(2\bar{2}),(1\bar{3})-(1\bar{2})\},\nonumber\\
&&\{(2\bar{1})-(2\bar{3}),(2\bar{2})-(1\bar{2}),(1\bar{1})-(1\bar{3})\}. 
\end{eqnarray}

Picking a context from each of the $3$ pairs of deterministic contexts and 
a context from the 3-hypercycle $(3\bar{1})-(3\bar{2})-(3\bar{3})-(3\bar{1})$, we need to check if such a set of $4$ contexts forms a MISC:
by verifying that it does not appear as a subset of the deterministic set of 6 contexts fixed by any of the remaining 
$3$-hypercycles. An example of such a set of 4 contexts is $\{(3\bar{1}),(2\bar{1}),(1\bar{2}),(1\bar{3})\}$ which is a ${\rm MISC(4)}$.
It is not a subset of any of the deterministic sets of contexts induced by 3-hypercycle extremal probabilistic models, namely:

\begin{eqnarray}
&&\{(2\bar{1}),(2\bar{2}),(2\bar{3}),(3\bar{1}),(3\bar{2}),(3\bar{3})\}\nonumber\\
&&\text{induced by } (1\bar{1})-(1\bar{2})-(1\bar{3})-(1\bar{1}),\nonumber\\
&&\{(1\bar{1}),(1\bar{2}),(1\bar{3}),(3\bar{1}),(3\bar{2}),(3\bar{3})\}\nonumber\\
&&\text{induced by } (2\bar{1})-(2\bar{2})-(2\bar{3})-(2\bar{1}),\nonumber\\
&&\{(1\bar{1}),(1\bar{2}),(1\bar{3}),(2\bar{1}),(2\bar{2}),(2\bar{3})\}\nonumber\\
&&\text{induced by } (3\bar{1})-(3\bar{2})-(3\bar{3})-(3\bar{1}),\nonumber\\
&&\{(1\bar{2}),(2\bar{2}),(3\bar{2}),(1\bar{3}),(2\bar{3}),(3\bar{3})\}\nonumber\\
&&\text{induced by } (1\bar{1})-(2\bar{1})-(3\bar{1})-(1\bar{1}),\nonumber\\
&&\{(1\bar{1}),(2\bar{1}),(3\bar{1}),(1\bar{3}),(2\bar{3}),(3\bar{3})\}\nonumber\\
&&\text{induced by } (1\bar{2})-(2\bar{2})-(3\bar{2})-(1\bar{2}),\nonumber\\
&&\{(1\bar{1}),(2\bar{1}),(3\bar{1}),(1\bar{2}),(2\bar{2}),(3\bar{2})\}\nonumber\\
&&\text{induced by } (1\bar{3})-(2\bar{3})-(3\bar{3})-(1\bar{3}).
\end{eqnarray}

In all, there are $9$ such ${\rm MISC(4)}$ and they are irreducible, i.e., no proper subset of these $9$ MISCs forms a MISC.
This is easy to verify, for example, for the ${\rm MISC(4)}$ $\{(3\bar{1}),(2\bar{1}),(1\bar{2}),(1\bar{3})\}$: every proper subset of this ${\rm MISC(4)}$
appears in one of the six deterministic sets of contexts. These irrMISCs are depicted in Fig.~\ref{irrMISC4} and listed below:

\begin{eqnarray}
&&\{(1\bar{1}),(2\bar{1}),(3\bar{2}),(3\bar{3})\}\nonumber\\
&&\{(1\bar{1}),(3\bar{1}),(2\bar{2}),(2\bar{3})\}\nonumber\\
&&\{(1\bar{2}),(2\bar{2}),(3\bar{1}),(3\bar{3})\}\nonumber\\
&&\{(1\bar{2}),(3\bar{2}),(2\bar{1}),(2\bar{3})\}\nonumber\\
&&\{(1\bar{3}),(2\bar{3}),(3\bar{1}),(3\bar{2})\}\nonumber\\
&&\{(1\bar{3}),(3\bar{3}),(2\bar{1}),(2\bar{2})\}\nonumber\\
&&\{(2\bar{1}),(3\bar{1}),(1\bar{2}),(1\bar{3})\}\nonumber\\
&&\{(2\bar{2}),(3\bar{2}),(1\bar{1}),(1\bar{3})\}\nonumber\\
&&\{(2\bar{3}),(3\bar{3}),(1\bar{1}),(1\bar{2})\}.
\end{eqnarray}

\begin{figure}
\centering
 \includegraphics[scale=0.33]{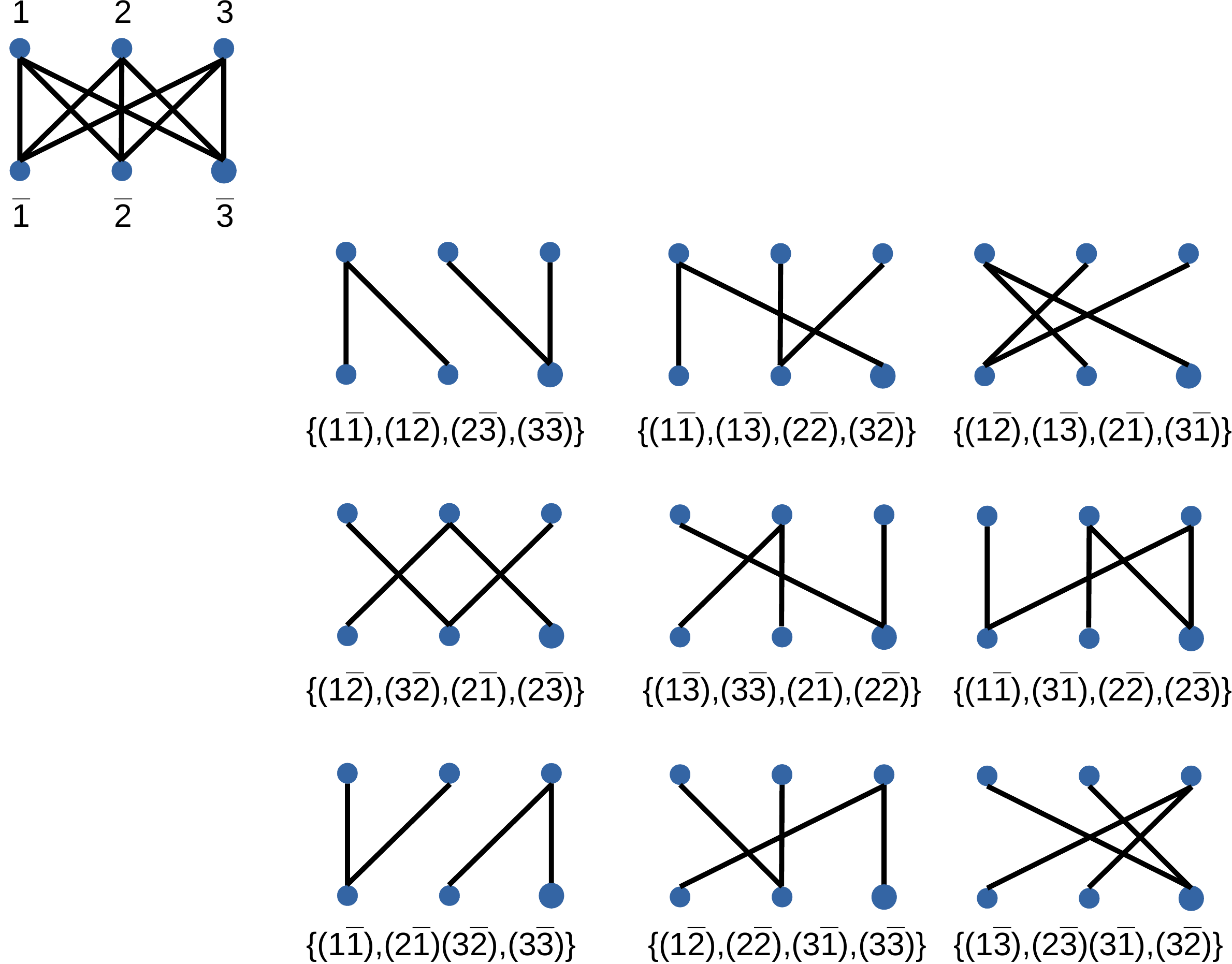}
 \caption{All irrMISC(4) for the 2Reg$(K_{3,3})$ scenario.}
 \label{irrMISC4}
\end{figure}
Each of these irrMISC(4)s (with uniform $q$) corresponds to a noncontextuality inequality:
\begin{eqnarray}\label{eq:k33irrMISC4}
&&{\rm Corr}_{\rm irrMISC(4)}\nonumber\\
&\equiv&\frac{1}{4}\sum_{i\in {\rm irrMISC(4)}}\sum_{x=1}^4p(m_i=x,s_i=x|M_i,S_i)\nonumber\\
&\leq&\max_{\lambda\in\Lambda}\frac{1}{4}\sum_{i\in {\rm irrMISC(4)}}\zeta(M_i,\lambda)\nonumber\\
&\leq&\frac{1}{4}\left(3+\frac{1}{2}\right)=\frac{7}{8}.
\end{eqnarray}

Are there any still smaller MISCs, say ${\rm MISC(3)}$, in this contextuality scenario?  
Indeed, such MISCs exist and they correspond precisely to the perfect matchings of the graph $K_{3,3}$. 
Each of the six vertices of $K_{3,3}$ is an origin of a 3-hypercycle (corresponding to 2Reg$(K_{1,3})$; see Fig.~\ref{3hypercycles}) in the contextuality 
scenario 2Reg$(K_{3,3})$. Hence, a perfect matching -- namely, a set of disjoint edges such that they cover all the six 
vertices of the graph -- ensures that the three hyperedges corresponding to these edges in the perfect matching cannot all be made deterministic by any 3-hypercycle 
extremal probabilistic model on 2Reg$(K_{3,3})$. This is because at least one of the three hyperedges, e.g. $\{(1\bar{1}),(2\bar{2}),(3\bar{3})\}$,
will be indeterministic (forming a part of a 3-hypercycle) in these extremal probabilistic models. Indeed, these three hyperedges can't be made deterministic 
by any extremal probabilistic model at all, 
since all extremal probabilistic models on 2Reg$(K_{3,3})$ are induced by odd hypercycles and the remaining extremal probabilistic models must therefore 
contain at least a 5-hypercycle. A 5-hypercycle extremal probabilistic model would make 4 contexts deterministic, but these 4 contexts will 
come in pairs that each share a deterministic node assigned probability 1. This means a maximum of 2 independent deterministic 
contexts in any other extremal probabilistic models besides those induced by 3-hypercycles: hence these extremal probabilistic models cannot make more than 
two of the three contexts in a perfect matching deterministic (since the three contexts share no nodes in 2Reg$(K_{3,3})$).
There are six perfect matchings of $K_{3,3}$, hence 6 instances of ${\rm MISC(3)}$, all of which are in fact irreducible:
\begin{eqnarray}
&&\{(1\bar{1}),(2\bar{2}),(3\bar{3})\}\nonumber\\
&&\{(1\bar{1}),(2\bar{3}),(3\bar{2})\}\nonumber\\
&&\{(2\bar{2}),(1\bar{3}),(3\bar{1})\}\nonumber\\
&&\{(3\bar{3}),(1\bar{2}),(2\bar{1})\}\nonumber\\
&&\{(1\bar{3}),(2\bar{1}),(3\bar{2})\}\nonumber\\
&&\{(3\bar{1}),(1\bar{2}),(2\bar{3})\}.
\end{eqnarray}
See Fig.~\ref{irrMISC3}. Each of these irrMISC(3)s yields a noncontextuality inequality (again, assuming uniform $q$ here):
\begin{eqnarray}\label{eq:k33irrMISC3}
&&{\rm Corr}_{\rm irrMISC(3)}\nonumber\\
&\equiv&\frac{1}{3}\sum_{i\in {\rm irrMISC(3)}}\sum_{x=1}^4p(m_i=x,s_i=x|M_i,S_i)\nonumber\\
&\leq&\max_{\lambda\in\Lambda}\frac{1}{3}\sum_{i\in {\rm irrMISC(3)}}\zeta(M_i,\lambda)\nonumber\\
&\leq&\frac{1}{3}\left(2+\frac{1}{2}\right)=\frac{5}{6}.
\end{eqnarray}

\begin{figure}
\centering
 \includegraphics[scale=0.33]{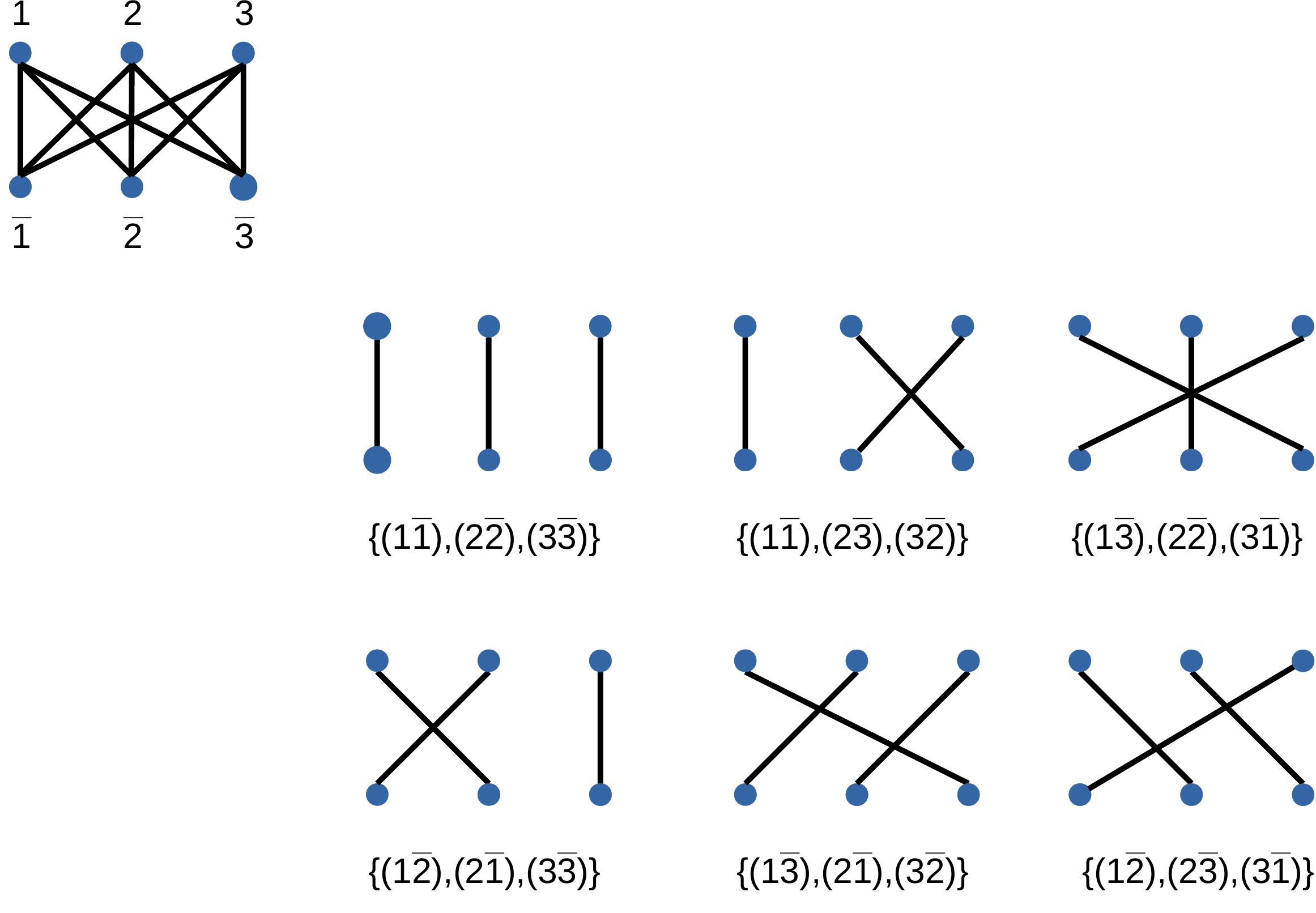}
 \caption{All irrMISC(3) for the 2Reg$(K_{3,3})$ scenario.}
 \label{irrMISC3}
\end{figure}

The noncontextuality inequality of Ref.~\cite{KunjSpek} can then be obtained by coarse-graining these irrMISC(3) inequalities, say the 
ones corresponding to irrMISCs $\{(1\bar{1}),(2\bar{3}),(3\bar{2})\}$, $\{(2\bar{2}),(1\bar{3}),(3\bar{1})\}$ and 
$\{(3\bar{3}),(1\bar{2}),(2\bar{1})\}$ or the ones corresponding to
irrMISCs $\{(1\bar{1}),(2\bar{2}),(3\bar{3})\}$, $\{(1\bar{3}),(2\bar{1}),(3\bar{2})\}$, and $\{(3\bar{1}),(1\bar{2}),(2\bar{3})\}$, 
to yield:

\begin{eqnarray}
&&A\equiv\frac{1}{9}\sum_{i,j=1}^3\sum_{x=1}^4p(m_{(i\bar{j})}=x,s_{(i\bar{j})}=x|M_{(i\bar{j})},S_{(i\bar{j})})\nonumber\\
&\leq&\max_{\lambda\in\Lambda}\frac{1}{9}\sum_{i,j=1}^3\zeta(M_{(i\bar{j})},\lambda)\nonumber\\
&\leq&\frac{1}{9}\left(6+3.\frac{1}{2}\right)=\frac{5}{6}.
\end{eqnarray}

Note that the upper bounds on the average correlation corresponding to irrMISC(3)s and irrMISC(4)s were first obtained in Ref.~\cite{anithesis} via 
an implementation of Fourier-Motzkin elimination.\footnote{See \cite{anithesis} for details of that numerical approach. 
It consists of writing down all the positivity and normalization constraints on the average correlation for each context and then 
eliminating the ontological variables via Fourier-Motzkin elimination to eventually yield operational noise-robust noncontextuality inequalities.}

On the other hand, our derivation hinges on a conceptual insight -- based on the mapping 2Reg($\cdot$) and Theorem \ref{extremals} --
that clarifies why we expect the average source-measurement correlation of particular sets of contexts (rather than arbitrary sets of contexts) in these 
noncontextuality inequalities to be bounded away from 1. It boils down to identifying MISCs and irrMISCs in a contextuality scenario.
Indeed, as we now show, our understanding lets us obtain previously undiscovered 
noncontextuality inequalities in other KS-uncolourable contextuality scenarios.

\subsubsection{2Reg$(K_{1,7})$}

We denote the edges of $K_{1,7}$ (and the corresponding contexts in 2Reg$(K_{1,7})$) by 
$$\{(1\bar{1}),(1\bar{2}),(1\bar{3}),(1\bar{4}),(1\bar{5}),(1\bar{6}),(1\bar{7})\}.$$

Since every edge is connected to every other edge in $K_{1,7}$ and vertex $1$ is the origin of all hypercycles,
we have that each choice of a set of 3 contexts in 
2Reg$(K_{1,7})$ will form a 3-hypercycle. Extremal probabilistic models induced by subscenarios containing these 3-hypercycles 
can make at most all the remaining 4 contexts deterministic. Indeed, taking out 3 edges from $K_{1,7}$ yields 
$K_{1,4}$ as a remnant and 2Reg$(K_{1,4})$ admits only deterministic extremal probabilistic models.

Each MISC of size $c$ would require a set of $c$ contexts such that no more than $c-1$ of them 
can be made deterministic in {\em any} extremal probabilistic model on 2Reg$(K_{1,7})$. 
We know that every choice of a set of 
4 contexts in 2Reg$(K_{1,7})$ can be made deterministic by some extremal probabilistic model since every such choice is in one-to-one 
correspondence with a choice of a 3-hypercycle (consisting of the remaining 3 contexts) 
inducing such an extremal probabilistic model on 2Reg$(K_{1,7})$: we have $^7C_4 = ^7C_3=35$ such choices. 
Hence a set of contexts of size $\leq 4$ can never form a MISC: there will always exist an extremal probabilistic model
which will make all of the contexts in the set deterministic. All irrMISCs are therefore 
of size $c=5$ in 2Reg$(K_{1,7})$ and every set of 5 contexts ($n-k+1=7-3+1=5$) forms an irrMISC(5). See Fig.~\ref{irrMISC5}.

\begin{figure}
\centering
 \includegraphics[scale=0.33]{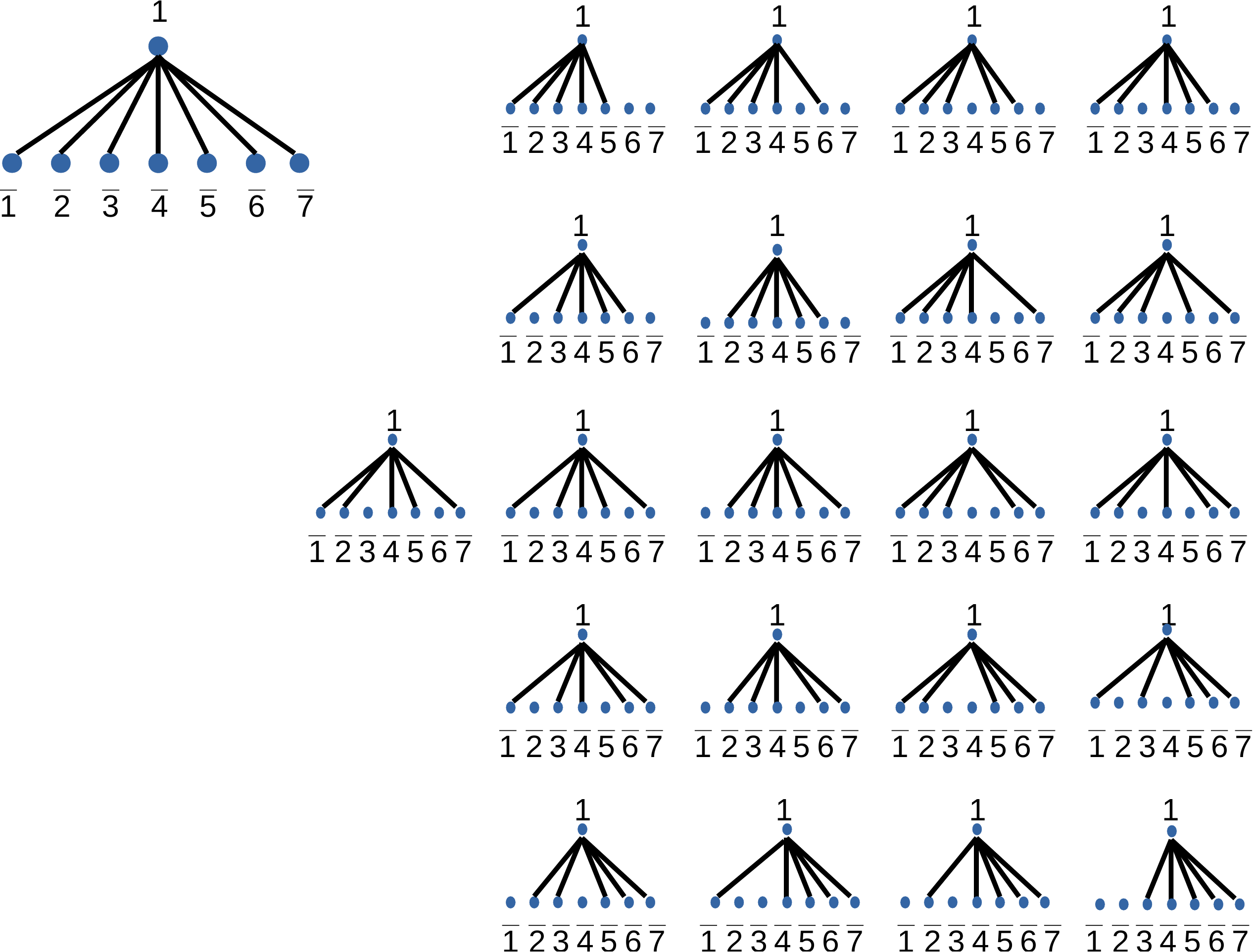}
 \caption{All irrMISC(5) for the 2Reg$(K_{1,7})$ scenario.}
 \label{irrMISC5}
\end{figure}

The noise-robust noncontextuality inequalities for 2Reg$(K_{1,7})$ then correspond to average source-measurement correlation (assuming uniform $q$) for:
\begin{enumerate}
 \item irrMISC(5)s:
 \begin{eqnarray}\label{eq:k17irrMISC5}
  &&{\rm Corr}_{\rm irrMISC(5)}\nonumber\\
  &\equiv&\frac{1}{5}\sum_{j\in {\rm irrMISC(5)}}\sum_{x=1}^6p(m_j=x,s_j=x|M_j,S_j)\nonumber\\
  &\leq& \max_{\lambda\in\Lambda}\frac{1}{5}\sum_{j\in {\rm irrMISC(5)}}\zeta(M_j,\lambda)\nonumber\\
  &=&\frac{1}{5}\left(4+\frac{1}{2}\right)\nonumber\\
  &=&\frac{9}{10}.
 \end{eqnarray}
There are $^7C_5=21$ such inequalities.

\item Each irrMISC(5) and 1 indeterministic context:
\begin{eqnarray}
 &&{\rm Corr}_{\rm MISC(6)}\nonumber\\
 &\equiv&\frac{1}{6}\sum_{j=1}^6\sum_{x=1}^6\zeta(m_{r_j}=x,s_{r_j}=x|M_{r_j},S_{r_j})\nonumber\\
 &\leq& \max_{\lambda\in\Lambda}\frac{1}{6}\sum_{j=1}^6\zeta(M_{r_j},\lambda)\nonumber\\
 &=&\frac{1}{6}\left(4+2.\frac{1}{2}\right)\nonumber\\
 &=&\frac{5}{6}.
\end{eqnarray}
for every subset of 6 distinct contexts $\{r_1,r_2,r_3,r_4,r_5,r_6\}\subset\{(1\bar{1}),(1\bar{2}),\dots,(1\bar{7})\}$.
There are $^7C_6=7$ such inequalities.
 
 \item All the contexts (or each irrMISC(5) and 2 indeterministic contexts):
\begin{eqnarray}
 &&{\rm Corr}_{\rm MISC(7)}\nonumber\\
 &\equiv&\frac{1}{7}\sum_{i=1}^7\sum_{x=1}^6p(m_i=x,s_i=x|M_i,S_i)\nonumber\\
 &\leq& \max_{\lambda\in\Lambda}\frac{1}{7}\sum_{i=1}^7\zeta(M_i,\lambda)\nonumber\\
 &=&\frac{1}{7}\left(4+3.\frac{1}{2}\right)\nonumber\\
 &=&\frac{11}{14}.
\end{eqnarray}
 There is one such inequality.
 \end{enumerate}

\subsubsection{2Reg$(K_{1,n})$ with $n (>1)$ odd}
We denote the edges of $K_{1,n}$ (and the corresponding contexts in 2Reg$(K_{1,n})$) by 
$$\{(1\bar{1}),(1\bar{2}),\dots,(1\bar{n})\}.$$

Since every edge is connected to every other edge in $K_{1,n}$, each triple of contexts in 
2Reg$(K_{1,n})$ will form a 3-hypercycle. Extremal probabilistic models induced by subscenarios containing these 3-hypercycles 
can {\em at most} make the remaining $n-3$ contexts deterministic. Indeed, taking out 3 edges from $K_{1,n}$ yields 
$K_{1,n-3}$ as a remnant and 2Reg$(K_{1,n-3})$ does admit deterministic extremal probabilistic models (from Theorem \ref{thmunc}, 
since $n-3$ is even for any odd $n>1$.)

Each MISC of size $c$ would require a set of $c$ contexts such that no more than $c-1$ of them 
can be made deterministic in {\em any} extremal probabilistic model on 2Reg$(K_{1,n})$. 
We know that every choice of a set of $n-3$ contexts in 2Reg$(K_{1,n})$ can be made deterministic by some 
extremal probabilistic model since every such choice is in one-to-one correspondence with a choice of a 3-hypercycle 
(consisting of the remaining 3 contexts) inducing such an extremal probabilistic model on 2Reg$(K_{1,n})$: $^nC_{n-3}= {^n}C_3=\frac{n!}{3!(n-3)!}$. 
Hence a set of contexts of size $\leq n-3$ can never form a MISC: there will always exist an extremal probabilistic model
which will make all of the contexts in the set deterministic. All irrMISCs are therefore 
of size $n-2$ in 2Reg$(K_{1,n})$ and every set of $n-2$ contexts forms an irrMISC$(n-2)$. Clearly, the sufficient condition for a set of contexts to be a MISC that we identified in Sec.~6.2.1 is also necessary for contextuality scenarios of the type 2Reg$(K_{1,n})$ for odd $n\geq3$.

The noncontextuality inequalities for 2Reg$(K_{1,n})$ then correspond to average source-measurement correlation for
\begin{enumerate}
 \item irrMISC$(n-2)$s:
\begin{eqnarray}\label{eq:k1nirrMISCn-2}
&&{\rm Corr}_{\rm irrMISC(n-2)}\nonumber\\
&\equiv&\frac{1}{n-2}\sum_{j\in {\rm irrMISC}(n-2)}\sum_{x=1}^{n-1}p(m_j=x,s_j=x|M_j,S_j)\nonumber\\
&\leq& \max_{\lambda\in\Lambda}\frac{1}{n-2}\sum_{j\in {\rm irrMISC}(n-2)}\zeta(M_j,\lambda)\nonumber\\
  &=&\frac{1}{n-2}\left(n-3+\frac{1}{2}\right)\nonumber\\
  &=&1-\frac{1}{2(n-2)}.
 \end{eqnarray}
There are $^nC_{n-2}=\frac{n!}{2!(n-2)!}=\frac{n(n-1)}{2}$ such inequalities.

\item Each irrMISC$(n-2)$ and 1 indeterministic context:
\begin{eqnarray}
&&{\rm Corr}_{\rm MISC(n-1)}\nonumber\\
&\equiv&\frac{1}{n-1}\sum_{j=1}^{n-1}\sum_{x=1}^{n-1}p(m_{r_j}=x,s_{r_j}=x|M_{r_j},S_{r_j})\nonumber\\
 &\leq& \max_{\lambda\in\Lambda}\frac{1}{n-1}\sum_{j=1}^{n-1}\zeta(M_{r_j},\lambda)\nonumber\\
 &=&\frac{1}{n-1}\left(n-3+2.\frac{1}{2}\right)\nonumber\\
 &=&1-\frac{1}{n-1}.
\end{eqnarray}
for every subset of $n-1$ distinct contexts $\{r_1,r_2,\dots,r_{n-1}\}\subset\{(1\bar{1}),(1\bar{2}),\dots,(1\bar{n})\}$.
There are $^nC_{n-1}=n$ such inequalities.
 
 \item All the contexts (or each irrMISC$(n-2)$ and 2 indeterministic contexts):
\begin{eqnarray}
 &&{\rm Corr}_{\rm MISC(n)}\nonumber\\
 &\equiv&\frac{1}{n}\sum_{i=1}^n\sum_{x=1}^{n-1}p(m_i=x,s_i=x|M_i,S_i)\nonumber\\
 &\leq& \max_{\lambda\in\Lambda}\frac{1}{n}\sum_{i=1}^n\zeta(M_i,\lambda)\nonumber\\
 &=&\frac{1}{n}\left(n-3+3.\frac{1}{2}\right)\nonumber\\
 &=&1-\frac{3}{2n}.
\end{eqnarray}
 There is one such inequality.\footnote{Note that the contextuality scenario 2Reg$(K_{1,5})$ appeared in Ref.~\cite{CabelloK15}, where a subnormalized assignment of 
quantum projectors to this scenario in $\mathbb{C}^6$ was presented. This set of projectors, however, is not a KS set because of the 
subnormalization.}
\end{enumerate}

\subsubsection{2Reg($K_{m,n}$), for odd $mn>1$}
We now extend the derivation of irrMISC noncontextuality inequalities above to the case of all  KS-uncolourable 2-regular scenarios, 2Reg($K_{m,n}$), obtained from arbitrary complete bipartite graphs $K_{m,n}$. These are just those $K_{m,n}$ with odd $mn>1$ (from Theorem \ref{thmunc}): $K_{3,3}$ and $K_{1,n}$ (odd $n>1$) 
are special cases of these, so the recipe for noncontextuality inequalities obtained here will recover the noncontextuality inequalities we have already obtained.

Obtaining these noncontextuality inequalities entails two things: identifying all the irrMISCs in the contextuality scenario and calculating their upper bounds due to noncontextuality, i.e.,  $\beta(\Gamma,q)$. Since these are 2-regular scenarios, calculating the upper bounds is easy (due to Theorem \ref{halfinteger}).

Before we proceed with the general result, we need the following definitions:

{\em Edge cover:} An edge cover of a graph is a set of its edges such that every vertex of the graph belongs to at least one of the edges in this set. For example, $K_{3,3}$ has an edge cover $\{\{1,\bar{1}\},\{1,\bar{2}\},\{2,\bar{2}\},\{2,\bar{3}\},\{3,\bar{3}\}\}$.

{\em Minimum edge cover:} An edge cover of the smallest possible size for a graph is called its minimum edge cover. Size of an edge cover is given by the number of edges it contains. For example, $K_{3,3}$ has a minimum edge cover $\{\{1,\bar{1}\},\{2,\bar{2}\},\{3,\bar{3}\}\}$.

{\em Minimal edge cover:} An edge cover such that no proper subset of it is an edge cover is called a minimal edge cover. Every minimum edge cover is minimal, but not conversely. For example, $K_{3,3}$ has a minimal edge cover $\{\{1,\bar{1}\},\{1,\bar{2}\},\{2,\bar{3}\},\{3,\bar{3}\}\}$. 

Below, we prove some properties of minimal edge covers of $K_{m,n}$ before moving on to a characterization of irrMISCs in $K_{m,n}$.

\begin{theorem}\label{generalKmnfacts}
Any minimal edge cover of $K_{m,n}$ partitions the vertices of $K_{m,n}$ into a disjoint union of $\kappa$ connected subgraphs, where $\kappa\in\{1,2,\dots,\min\{m,n\}\}$, and we have for the number of edges, $N$, in the minimal edge cover, $N=m+n-\kappa$. Hence, we have 
$$\max\{m,n\}\leq N\leq m+n-1.$$
When either $m=1$ or $n=1$, we have that $K_{m,n}$ is its own only minimal (hence also minimum) edge cover so that $\kappa=1$. For  $m,n\geq2$, the total number of minimum edge covers of $K_{m,n}$ is 
$$\frac{\max\{m,n\}!}{|m-n|!}(\min\{m,n\})^{|m-n|},$$
and the total number of minimal (not necessarily minimum) edge covers of $K_{m,n}$ ($m,n\geq2$) is
\begin{equation}
\sum_{N=\max\{m,n\}}^{m+n-2} (\textrm{number of minimal edge covers of size } N).
\end{equation}
\end{theorem}
\begin{proof}
	For a given $K_{m,n}$, let's call the set of $m$ vertices $S_m$ and the set of $n$ vertices $S_n$.	The size, $N$, of an edge cover of $K_{m,n}$ must satisfy 
	\begin{equation}
	N=\sum_{v\in S_m}{\rm deg}(v)=\sum_{v\in S_n}{\rm deg}(v),
	\end{equation}
	where $v$ denotes a vertex of $K_{m,n}$ and ${\rm deg}(v)$ denotes the degree of the vertex, i.e., the number of edges in which it appears. 
	
	Note that two vertices connected by an edge in a minimal edge cover cannot both have degree $>1$: if an edge cover has a pair of degree 2 vertices connected by an edge, then the edge cover cannot be minimal since the said connecting edge can be dropped while maintaining the edge cover property. (See Fig.~\ref{degreefig}.) This means that a minimal edge cover of $K_{m,n}$ is such that any vertex of degree 2 or more is only connected to degree 1 vertices, hence the minimal edge cover is a disjoint union of connected bipartite subgraphs of type $K_{1,b}$ or $K_{a,1}$, where $a\leq m, b\leq n$. Denoting the set of vertices of each subgraph by $V_i$, we have that number of edges in such a subgraph is $|V_i|-1$. Thus, the total number of edges in a minimal edge cover, $N=\sum_{i=1}^{\kappa}|V_i|-\kappa=m+n-\kappa$, where $\kappa$ is the number of disjoint subgraphs whose union yields the minimal edge cover. Clearly, $1\leq\kappa\leq \min\{m,n\}$, where $\kappa=1$ corresponds to the case of any $K_{m,n}$ graph with $m=1$ or $n=1$ since it is its own minimal edge cover and we have $N=m+n-1$. For any other $K_{m,n}$ (with $m,n\geq2$), we have that $\kappa\geq2$ and the maximum size of a minimal edge cover is $m+n-2$: this is achieved when a vertex $v_{\min}\in S_{\min\{m,n\}}$ is connected to all but one (say, $v_{\max}$) of the vertices in $S_{\max\{m,n\}}$. The remaining vertex $v_{\max}\in S_{\max\{m,n\}}$ is then connected to all vertices of $S_{\min\{m,n\}}$ except $v_{\min}\in S_{\min\{m,n\}}$. We then have $N=m+n-2$.
	
	The total number of minimum edge covers can be computed as follows: every vertex in $S_{\min\{m,n\}}$ is connected one-to-one via an edge to a vertex in $S_{\max\{m,n\}}$ and there are $\frac{\max\{m,n\}!}{|m-n|!}$ possible ways to do this. For each such way, each of the remaining $|m-n|$ vertices in $S_{\max\{m,n\}}$ can be connected via an edge to one of the $\min\{m,n\}$ vertices of $S_{\min\{m,n\}}$, so there are $(\min\{m,n\})^{|m-n|}$ possible configurations for the remaining edges. This yields a total of $\frac{\max\{m,n\}!}{|m-n|!}(\min\{m,n\})^{|m-n|}$ minimum edge covers for $K_{m,n}$.
\end{proof}

We leave the general case as an open question:

{\em What is the number of minimal-but-not-minimum edge covers for an arbitrary complete bipartite graph $K_{m,n}$, where $m,n\geq 2$?}

\begin{figure}\centering
	\includegraphics[scale=0.4]{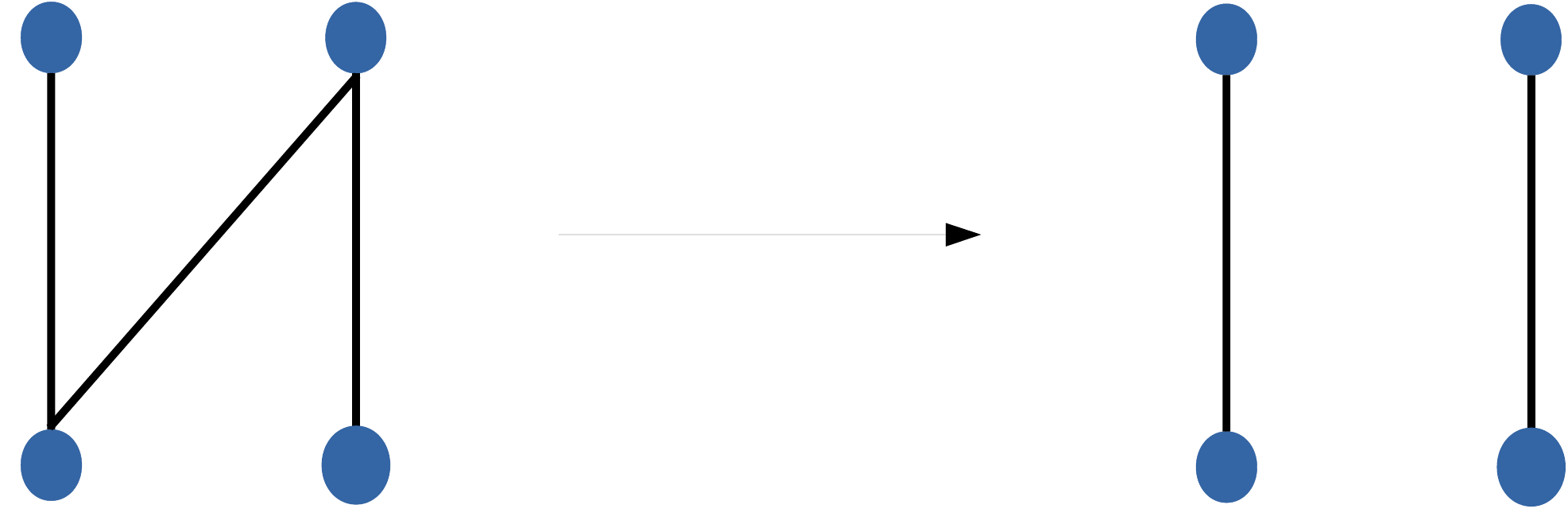}
	\caption{Why two vertices connected by an edge in a minimal edge cover of a bipartite graph cannot both have degree greater than 1.}
	\label{degreefig}
\end{figure}

\begin{theorem}\label{MISCKmn}
	
	A set of contexts in 2Reg$(K_{m,n})$, odd $mn>1$, forms a MISC if and only if 
	\begin{enumerate}
		\item for $m=1$ or $n=1$:
		the number of contexts in the set is {\em at least} $mn-2$.
		\item for $m\geq 3$ and $n\geq3$: the corresponding set of edges in $K_{m,n}$ is an edge cover of $K_{m,n}$.
	\end{enumerate}
\end{theorem}

\begin{proof}
	
	{\bf Case 1}, i.e., $m=1$ or $n=1$ (and odd $mn>1$):
	
	In this case, we have that every choice of 3 edges in $K_{m,n}$ is a claw and therefore induces a 3-hypercycle extremal probabilistic model on 2Reg($K_{m,n}$) that renders all the remaining ($mn-3$) contexts in 2Reg($K_{m,n}$) deterministic. To form a MISC, then, requires at least $mn-3+1=mn-2$ contexts. (See Fig.~\ref{irrMISC5} for a $K_{1,7}$ example.)
	
	{\bf Case 2}, i.e., $m\geq3$ and $n\geq3$ (and odd $mn>1$):
	
	{\em Every MISC of 2Reg$(K_{m,n})$ corresponds to an edge cover of $K_{m,n}$:} We show this by proving the contrapositive. If a set of edges is not an edge cover of $K_{m,n}$, then there exists a vertex in $K_{m,n}$ (not covered by the set of edges) that can support a claw (subgraph $K_{1,3}$ of $K_{m,n}$) which corresponds to a 3-hypercycle in 2Reg$(K_{m,n})$, odd $mn>1$. 
	All the contexts in the corresponding set of contexts can then be made deterministic relative to an extremal probabilistic model induced by this 3-hypercycle (cf.~Theorems \ref{halfinteger} and \ref{khypercyclemodels}), hence the set cannot be a MISC. (See Fig.~\ref{inducedsubscenario}.)

	{\em Every edge cover of $K_{m,n}$ corresponds to a MISC of 2Reg($K_{m,n}$):} If a set of edges forms an edge cover of $K_{m,n}$, then there does not exist any vertex in $K_{m,n}$ that can support a claw disjoint from the set of edges. Hence, it is not possible to find a 3-hypercycle extremal probabilistic model on 2Reg($K_{m,n}$) that makes all the contexts in the set deterministic: at least one of the contexts must belong to a 3-hypercycle in any extremal probabilistic model induced by such a hypercycle. The set of contexts must therefore be a MISC.\footnote{Recall that we need to restrict ourselves to 3-hypercycle extremal probabilistic models to identify MISCs in 2Reg($K_{m,n}$) scenarios: any bigger odd hypercycles would
    make even fewer contexts deterministic than a 3-hypercycle extremal probabilistic model and lead us to an artificially lower bound on the noncontextuality inequality; we have to give a noncontextual model as much leeway as mathematically possible to reproduce perfect predictability and thus find an upper bound that cannot be exceeded by {\em any} noncontextual ontological model, not merely those using extremal probabilistic models induced by 5 or higher odd hypercycles.} (See Fig.~\ref{3hypercycles})

\end{proof}

\begin{figure}\centering
	\includegraphics[scale=0.35]{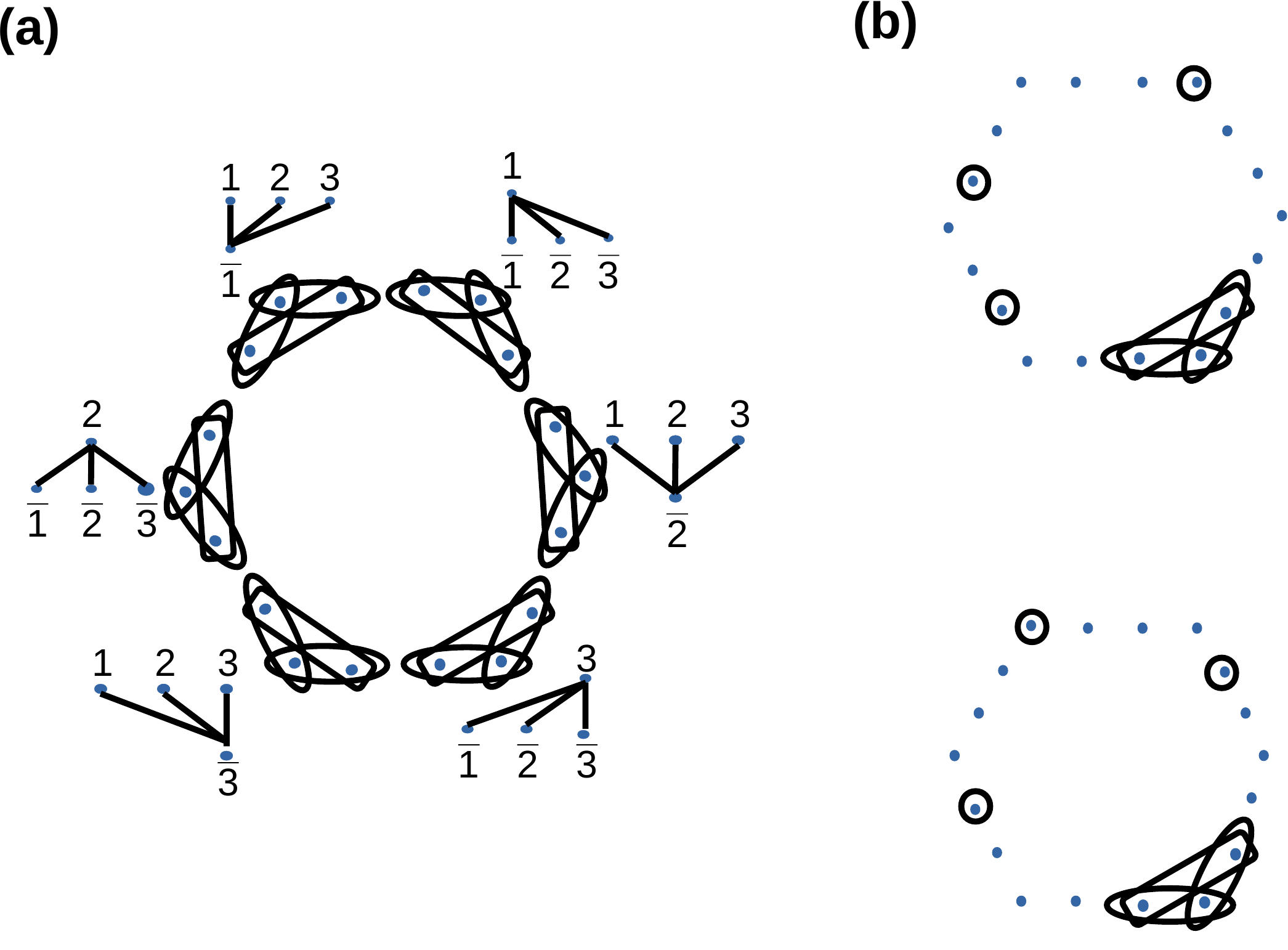}
	\caption{(a) All the six 3-hypercycles in the scenario 2Reg($K_{3,3}$), (b) Examples of extremal probabilistic models corresponding to a particular 3-hypercycle.}
	\label{3hypercycles}
\end{figure}

\begin{theorem}\label{irrMISCKmn}
A set of contexts is an irrMISC for 2Reg$(K_{m,n})$, odd $mn>1$, if and only if 
\begin{enumerate}
	\item for $m=1$ or $n=1$:
	the number of contexts in the set is {\em exactly} $mn-2$.
	\item for $m\geq 3$ and $n\geq3$: the corresponding set of edges in $K_{m,n}$ is a {\em minimal} edge cover of $K_{m,n}$, i.e., an edge cover such that none of its proper subsets is an edge cover.
\end{enumerate}
\end{theorem}
\begin{proof}
	This just follows from noting the definition of an irrMISC and Theorem \ref{MISCKmn}: an irrMISC is a MISC that does not contain another MISC as a proper subset. 
		
\end{proof}
\begin{figure}\centering
	\includegraphics[scale=0.35]{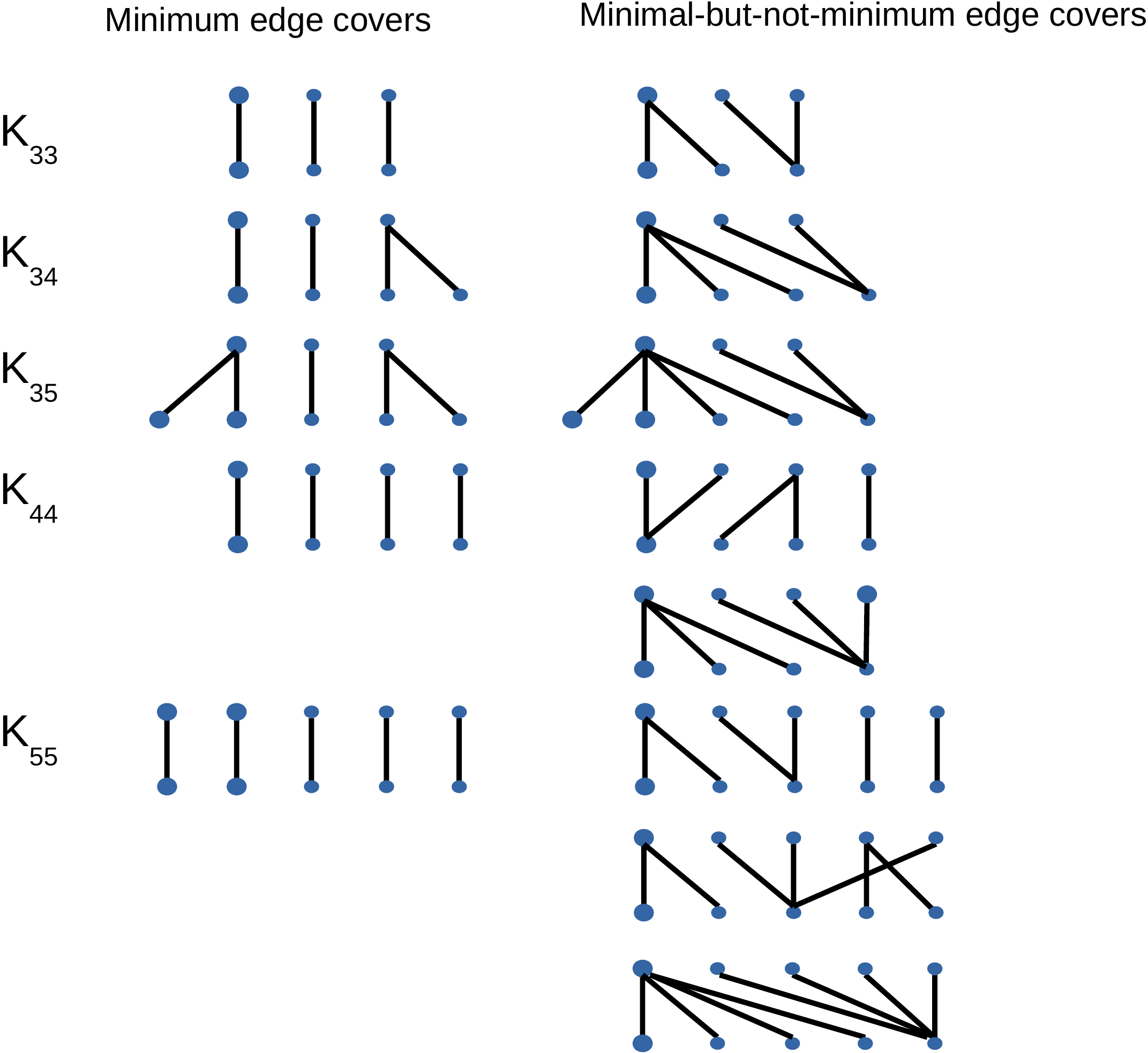}
	\caption{Representative minimal edge covers of various complete bipartite graphs, both minimum and non-minimum.}
	\label{edgecovers}

\end{figure}

Hence, a minimum edge cover always corresponds to an irrMISC, but not conversely. The 3-context irrMISCs of 2Reg($K_{3,3}$) form minimum edge covers (of size 3) but the 4-context irrMISCs of 2Reg($K_{3,3}$) don't form minimum edge covers. Thus, the smallest irrMISCs correspond to minimum edge covers, which for $K_{m,n}$ (odd $mn>1$) are of size $\max\{m,n\}$ (cf.~Theorem \ref{generalKmnfacts}). See Figs.~\ref{irrMISC4} and \ref{irrMISC3}.

When $m=n$, the smallest irrMISCs correspond to the perfect matchings\footnote{A {\em perfect matching} of a graph is a set of mutually disjoint edges that cover all vertices of the graph.} of $K_{m,m}$ and thus there are $m!$ such irrMISCs. Recall that for $K_{3,3}$, there are 6 such irrMISCs. The other irrMISCs are minimal edge covers that are not minimum. The number of these minimal-but-not-minimum edge covers is 9 for $K_{3,3}$: every minimal-but-not-minimum edge cover of $K_{3,3}$ should contain at least one vertex of degree 2 because no vertex can be degree 3 if the edge cover is minimal and the edge cover would be a minimum edge cover if all vertices are degree 1; there are three possible choices for a degree 2 vertex among $\{1,2,3\}$ and for each such choice there are three possible choices of pairs of vertices connected to it given by $\{\{\bar{1},\bar{2}\},\{\bar{2},\bar{3}\},\{\bar{1},\bar{3}\}$; choosing these fixes a minimal-but-not-minimum edge cover and we therefore have $3\times3=9$ irrMISCs arising from these edge covers. Note that no larger irrMISCs exist for $K_{3,3}$ since $N\leq m+n-2=4$. In Fig.~\ref{edgecovers}, we illustrate minimum and minimal-but-not-minimum edge covers for some classes of complete bipartite graphs.

In Table \ref{tabirrMISC}, we summarize facts about irrMISCs in various contextuality scenarios that we have considered in this section.

\begin{table*}
		\begin{tabular}{|l|c|r|}
			\toprule
			Scenario & Number of irrMISCs & Bounds on ${\rm Corr}$\\
			& & for uniform $q$\\
			\hline
			2Reg($K_{3,3}$) & $15$: $6$ $3$-context and $9$ $4$-context irrMISCs (Figs.~\ref{irrMISC4},\ref{irrMISC3})& Eqs.~\eqref{eq:k33irrMISC4}, \eqref{eq:k33irrMISC3}\\
			\hline
			2Reg($K_{1,7}$) & $21$: All $5$-context (Fig.~\ref{irrMISC5}) & Eq.~\eqref{eq:k17irrMISC5}\\
			\hline
			2Reg($K_{1,n}$) (with odd $n>1$) & $\frac{n(n-1)}{2}$: All $(n-2)$-context & Eq.~\eqref{eq:k1nirrMISCn-2}\\
			\hline
		\end{tabular}
		\caption{Summary of irrMISCs for a familar of $2$-regular contextuality scenarios. A general characterization of irrMISCs for contextuality scenarios of type 2Reg($K_{m,n}$) (with odd $mn>1$), of which all the examples above are instances, can be found in  Theorem \ref{irrMISCKmn}.}
\label{tabirrMISC}
\end{table*}

\section{Discussion and future work}
To summarize, we have presented a framework for noise-robust noncontextuality inequalities that are inspired by logical proofs of the Kochen-Specker theorem. We have identified special sets of these inequalities, corresponding to irreducible minimally indeterministic sets of contexts (or irrMISCs), that are independent of each other and can generate any other noise-robust noncontextuality inequality corresponding to a minimally indeterministic set of contexts (MISC) or even any other (non-MISC) set of contexts. The basic building blocks of any noise-robust noncontextuality inequality obtained in this framework are the noise-robust noncontextuality inequalities for irrMISCs. Along the way, we also obtained a parameterization of contextuality scenarios and identified ways to detect their KS-uncolourability.

Note that the contextuality scenarios we have considered within this framework are all required to be KS-uncolourable and this is the {\em only} restriction on them.\footnote{The case of KS-colourable contextuality scenarios that fit within a generalization of the CSW framework \cite{CSW} was considered in Ref.~\cite{robustcsw}. These KS-colourable scenarios satisfy the property that the set of probabilistic models satisfying consistent exclusivity on them coincides with the set of general probabilistic models on them. The case of KS-colourable contextuality scenarios where this property fails -- and which are therefore outside the purview of Refs.~\cite{CSW, robustcsw} -- will be taken up in future work.} In particular, we {\em do not} insist that they admit a realization with {\em KS sets} of projectors in quantum theory. When they do admit such a realization, we have that quantum theory can violate any noise-robust noncontextuality inequality by achieving ${\rm Corr}_q=1$ in the ideal noiseless limit. When they don't admit realization with KS sets then it remains an open question whether noise-robust noncontextuality inequalities of type ${\rm Corr}_q\leq \beta(\Gamma,q)$ might still admit a quantum violation using nonprojective positive operator-valued measures (POVMs).

We conclude with some directions for future work building on this framework.
\begin{enumerate}
	
	\item Sufficiency of the noise-robust noncontextuality inequalities from irrMISCs: 
	
	A crucial open question is whether the satisfaction of noncontextuality inequalities \`a la Eq.~\eqref{ncineqschematic} for the set of all irrMISCs of a given KS-uncolourable contextuality scenario is sufficient to guarantee the existence of a noncontextual ontological model in this scenario. In principle, one could check this in a brute-force manner for some scenario by using the algorithmic approach of Ref.~\cite{schmidWS} but, ideally, one would like an analytical proof of sufficiency (or otherwise) that applies to arbitrary KS-uncolourable contextuality scenarios.
	
	A weaker open question is whether it's possible to saturate any noncontextuality inequality of the form in Eq.~\eqref{ncineqschematic} in a noncontextual ontological model. We do not have an answer to this but we provide a possible approach to constructing such models in Appendix \ref{NCmodels}.
	
	\item POVM realizations when no KS sets exist: 
	
	Do there exist KS-uncolourable contextuality scenarios that do not admit KS sets but allow for a realization with POVMs in such a way that a noise-robust noncontextuality inequality can be violated? If so, construct some examples and determine the optimal POVM realization for them.
	
	Note that since POVMs can realize arbitrary joint measurability structures in quantum theory \cite{KHF}, it is conceivable that some such realizations could work for violating noise-robust noncontextuality inequalities of the type obtained in this paper even when the KS-uncolourable contextuality scenario concerned does not admit a realization with KS sets. 
	
	\item Lower-dimensional POVM realizations when KS sets exist:
	 
	For KS-uncolourable contextuality scenarios that do admit KS sets, find a realization with POVMs that works on a Hilbert space of lower dimension than the KS set realization and still violates a noise-robust noncontextuality inequality. Is it possible, for example, to realize the 18 ray scenario \cite{Cabello18ray} (Fig.~\ref{fig2}) with four-outcome qubit POVMs such that an irrMISC noncontextuality inequality is violated?
	
	\item Other KS constructions and extremal probabilistic models on their contextuality scenarios:
	
	We have only considered KS-uncolourable contextuality scenarios of type 2Reg($K_{m,n}$) in detail in this paper. It would be interesting to analyze the original Kochen-Specker construction \cite{KS67} and others that do not fall within the family of contextuality scenarios we have considered here. In particular, the parameterization of KS-uncolourable scenarios in Section 4 should be useful in attempting such analyses. Note that our analysis in this paper relied heavily on the characterization of extremal probabilistic models on 2Reg($K_{m,n}$) contextuality scenarios given by Theorem \ref{halfinteger}. More generally, it relied on the characterization of extremal probabilistic models on arbitrary contextuality scenarios given by AFLS \cite{AFLS} (see Theorem \ref{extremals}). The characterization of Theorem \ref{extremals} will be useful in attempting to study contextuality scenarios that do not fall under the purview of Theorem \ref{halfinteger}. Indeed, it would be worthwhile to obtain a characterization of contextuality scenarios that admit unique probabilistic models since these are the ones that induce extremal probabilistic models on any contextuality scenario following Theorem \ref{extremals}.
	
	\item Applications to quantum information protocols: 
	
	Given that the framework we have proposed leverages KS-uncolourability to provide noise-robust operational signatures of contextuality, there is good reason to expect that it might be relevant for quantum information tasks. In particular, it could be used to study the question of how logical proofs of the KS theorem can be leveraged to provide advantages in (possibly some variant of) state discrimination \`a la Ref.~\cite{SchmidSpek} (which does not use contextuality arising from Kochen-Specker proofs). In particular, the task of maximizing the average source-measurement correlation (the quantity ${\rm Corr}_q$) could possibly be related to minimum error state discrimination as follows: 
		for any $d$-uniform contextuality scenario with $n$ contexts, we consider $n$ ensembles of states (each denoted by source setting $S_i$) such that the average preparation procedures associated with them ($[\top|S_i]$) are all operationally equivalent and we apply the assumption of preparation noncontextuality relative to this operational equivalence. Our task then is to discriminate between elements of each ensemble under an additional constraint on the set of allowed measurements, namely, that they satisfy the operational equivalences required for KS-uncolourability and we apply measurement noncontextuality relative to these operational equivalences. Maximizing ${\rm Corr}_q$ then corresponds to maximizing the average success probability of discrimination given by ${\rm Corr}_q$ for the set of ensembles in the support of probability distribution $q$ (defined over the set of $n$ contexts). If the value of ${\rm Corr}_q$ exceeds the bound from our noise-robust noncontextuality inequality, we then have that the operational theory allows a greater success probability for minimum error state discrimination than a theory which admits a noncontextual ontological model.
	
	Note also that the possibility of realizing violations of our noise-robust noncontextuality inequalities using POVMs on lower-dimensional systems (than the ones on which KS sets exist) also makes it interesting to study this problem from such a perspective.
	
\end{enumerate}

\section*{Acknowledgement}
I would like to thank Anirudh Krishna, Rob Spekkens, Tom\'a\v s Gonda, Tobias Fritz, Andreas Winter, Giulio Chiribella, Matt Pusey, and Nuriya Nurgalieva for discussions and feedback at various stages during the writing of this paper. Thanks are also due to two anonymous reviewers for a critical reading of the paper, in particular Reviewer 1 for pointing out possible connections with the sheaf-theoretic literature on contextuality. This research was supported by
Perimeter Institute for Theoretical Physics. Research at Perimeter Institute is supported by the
Government of Canada through the Department of Innovation, Science, and Economic Development
Canada and by the Province of Ontario through the Ministry of Research, Innovation and Science.
\appendix 
\appendixpage
\addappheadtotoc
	\section{Note on matching scenarios}\label{matchingscenarios}
	Here we elaborate on the connection between the contextuality scenario 2Reg$(G)$ obtained from a graph $G$ and the matching scenario of ${\rm L}(G)$, the line graph of $G$, cf.~Ref.~\cite{AFLS}.
	
	The line graph ${\rm L}(G)$ of a graph $G$ is obtained by representing each edge of $G$ as a vertex of $L(G)$ and for each pair of edges in $G$ that have a non-empty intersection we connect the corresponding vertices in $L(G)$ by an edge.
	
	Following Ref.~\cite{AFLS}, we then have:
	\begin{theorem}
		2Reg$(G)={\rm Mat}({\rm L}(G))$, i.e., the construction of a contextuality scenario 
		from a graph $G$ under the mapping 2Reg($\cdot$) is equivalent to the construction of a matching scenario 
		obtained from the line graph of $G$, namely, ${\rm L}(G)$. 
		Also, ${\rm Dual}(2Reg(G))={\rm L}(G)$, that is, the dual graph of 2Reg$(G)$ is the line graph of $G$.
	\end{theorem}
	
	A matching scenario ${\rm Mat}({\rm L}(G))$ is constructed as follows: every edge of ${\rm L}(G)$ is represented by a vertex of ${\rm Mat}({\rm L}(G))$ and each hyperedge of ${\rm Mat}({\rm L}(G))$ contains those vertices of ${\rm Mat}({\rm L}(G))$ which represent edges of ${\rm L}(G)$ that have a non-empty intersection. We leave it as an exercise for the reader to verify that the composition of the two mappings ${\rm L}(.)$ and ${\rm Mat}(.)$ applied to $G$ in that order as ${\rm Mat}\circ {\rm L}(G)$ is the same as the mapping 2Reg($G$) for any $G$.
	
	Hence, 2Reg$(G)$ is a matching scenario of ${\rm L}(G)$. Matching scenarios were discussed in \cite{AFLS}. 
	Indeed, the following lemma 
	shows how matching scenarios obtained corresponding to complete graphs discussed in Ref.~\cite{AFLS}
	arise from bipartite graphs under the mapping 2Reg($\cdot$).
	\begin{lemma}
		${\rm L}(K_{1,n})$ is the complete graph $K_n$. 2Reg$(K_{1,n})={\rm Mat}({\rm L}(K_{1,n}))$ is therefore the matching scenario ${\rm Mat}_n$
		of AFLS \cite{AFLS}.
	\end{lemma}
	We refer the reader to Sec.~9.4 of Ref.~\cite{AFLS} for further details on matching scenarios. We mention the connection to matching scenarios of L$(G)$ here only for completeness and, as such, they are not discussed in much detail in the main text.
	We infer properties of 2Reg$(G)$
	directly from $G$ instead of considering an intermediate graph ${\rm L}(G)$,
	since $G$ has a much more compact representation than $L(G)$ and the edge-hyperedge
	correspondence between $G$ and 2Reg$(G)$ can be exploited to understand the structure of 2Reg$(G)$ and possible probabilistic models on it. This in turn let us obtain our noncontextuality inequalities for such scenarios.

\section{Can the noise-robust noncontextuality inequalities of Eq.~\eqref{ncineqschematic} be saturated?}\label{NCmodels}
It remains an open question whether all noise-robust noncontextuality inequalities of the form in Eq.~\eqref{ncineqschematic} can be saturated by a noncontextual ontological model. We will, however, describe below a procedure for constructing a noncontextual ontological model saturating these inequalities {\em when certain conditions are satisfied.} (Specifically, Eq.~\eqref{phenconstraint} below.)\footnote{We do not know of concrete examples where these conditions might be satisfied but we consider it worthwhile to mention how a noncontextual model could be constructed, should the reader be interested in investigating the feasibility of our proposal for particular cases. Of course, a better solution would be to come up with a different recipe for a construction that works for any inequality of the form in Eq.~\eqref{ncineqschematic}.}

For each noise-robust noncontextuality inequality obtained from a KS-uncolourable contextuality scenario, given by
\begin{eqnarray}
&&{\rm Corr}_q \equiv  \sum_{i=1}^nq_i\sum_{x=1}^dp(m_i=x,s_i=x|M_i,S_i)\nonumber\\
&\leq& \beta(\Gamma,q),
\end{eqnarray}
we can sometimes construct a straightforward noncontextual ontological model that saturates it as long as some conditions are satisfied. If it exists, such a model would obviously also satisfy (if not saturate) any other noise-robust noncontextuality inequality. To see how this straightforward construction works, we first write ${\rm Corr}_q$ in terms of an ontological model as 

\begin{align}
&{\rm Corr}_q\nonumber\\
=&\sum_{i=1}^nq_i\sum_{x=1}^d\sum_{\lambda\in\Lambda}\xi(m_i=x|M_i,\lambda)\mu(s_i=x|S_i,\lambda)\mu(\lambda|S_i),\nonumber\\
=&\sum_{i=1}^nq_i\sum_{x=1}^d\sum_{\lambda\in\Lambda}\xi(m_i=x|M_i,\lambda)\mu(s_i=x|S_i,\lambda)\nu(\lambda),\nonumber\\
(&\textrm{using preparation noncontextuality, }\nonumber\\ &\mu(\lambda|S_i)=\nu(\lambda)\quad\forall\lambda, i).
\end{align}

We note the following constraint (independent of noncontextuality) on the ontological representation of preparations that needs to hold in the model:
	\begin{align}
	\mu(s_i|S_i,\lambda)\nu(\lambda)=\mu(\lambda|S_i,s_i)p(s_i|S_i).
	\end{align}
	This constraint implies the following phenomenological constraint on the ontological model:
	\begin{align}\label{phenconstraint}
	\forall [s_i|S_i]: \sum_{\lambda\in\Lambda}\mu(s_i|S_i,\lambda)\nu(\lambda)=p(s_i|S_i).
	\end{align}

The first step in constructing the model is to choose $\mu(s_i|S_i,\lambda)=\delta_{s_i,y}$, where $y$ is such that $$\max_{m_i}\xi(m_i|M_i,\lambda)=\xi(m_i=y|M_i,\lambda)\equiv\zeta(M_i,\lambda).$$

We do this for every $S_i, i\in\{1,2,\dots,n\}$, so that we now have 
\begin{equation}
{\rm Corr}_q=\sum_{\lambda\in\Lambda}\left(\sum_{i=1}^nq_i\zeta(M_i,\lambda)\right)\nu(\lambda).
\end{equation}
Note that in all this, the response functions satisfy the assumption of measurement noncontextuality.

Let us denote by $\Lambda_{\rm detP}$ the following set of $\lambda\in\Lambda$:
	\begin{equation}
	\Lambda_{\rm detP}\equiv\{\lambda\in\Lambda\big|\mu(s_i|S_i,\lambda)=1 \textrm{ for some } [s_i|S_i]\},
	\end{equation}
	where $s_i\in [d], i\in[n]$ for any $[s_i|S_i]$.

Next, we choose $\nu(\lambda)$ in such a way that the inequality is saturated: in particular, we take $\nu(\lambda)$ to be supported entirely over the set $$\Lambda_{\max}\equiv\{\lambda\in\Lambda\big| \sum_{i=1}^nq_i\zeta(M_i,\lambda)=\beta(\Gamma,q)\},$$
so that $\nu(\lambda)=0$ for all $\lambda\in\Lambda\backslash\Lambda_{\max}$. Recalling the constraint of Eq.~\eqref{phenconstraint}, that
	\begin{align}
	\forall [s_i|S_i]: \sum_{\lambda\in\Lambda}\mu(s_i|S_i,\lambda)\nu(\lambda)=p(s_i|S_i),
	\end{align}
	we have that our choice of $\mu(s_i|S_i,\lambda)$ and $\nu(\lambda)$ above should satisfy the property that
	\begin{equation}\label{weakcondn}
	\Lambda_{\rm detP}\cap \Lambda_{\max}\neq \varnothing,
	\end{equation}
	since otherwise we cannot satisfy the constraint of Eq.~\eqref{phenconstraint}. If there exist no $\lambda$ that lie in $\Lambda_{\rm detP}\cap\Lambda_{\max}$, then we cannot proceed with our construction of the noncontextual ontological model. If Eq.~\eqref{weakcondn} is satisfied, then our choices of $\mu(s_i|S_i,\lambda)$ and $\nu(\lambda)$ must further satisfy the constraint of Eq.~\eqref{phenconstraint}. Again, if this cannot be done, our construction cannot proceed.

Assuming we have satisfied the constraint of Eq.~\eqref{phenconstraint} while choosing $\mu(s_i|S_i,\lambda)$ and $\nu(\lambda)$ as prescribed above, we then have
\begin{equation}
{\rm Corr}_q=\beta(\Gamma,q)
\end{equation}
and the noise-robust noncontextuality inequality is thus saturated.\footnote{What we haven't shown above is a concrete example of a scenario where a noncontextual ontological model following our prescription can actually be constructed and, to that extent, we haven't contributed much to proving saturation. We have merely suggested a possible route towards it.}

\end{document}